\documentclass[journal,onecolumn,12pt]{IEEEtran}
\date{}

\usepackage{mathtools}

\usepackage[dvipsnames]{xcolor}

\usepackage{amsmath}
\usepackage{mathtools}
\newcommand\coolover[2]{\mathrlap{\smash{\overbrace{\phantom{%
    \begin{matrix} #2 \end{matrix}}}^{\mbox{$#1$}}}}#2} 

\newcommand\coolunder[2]{\mathrlap{\smash{\underbrace{\phantom{%
    \begin{matrix} #2 \end{matrix}}}_{\mbox{$#1$}}}}#2}

\usepackage{amsfonts}
\usepackage{amsmath}
\usepackage{amssymb}
\usepackage{amsthm}
\usepackage{array}
\usepackage{caption}
\usepackage{cite}
\usepackage{citesort}
\usepackage{color}
\usepackage{epsfig}
\usepackage{enumitem} 
\usepackage{flexisym}
\usepackage{graphicx}
\usepackage{hyphenat}
\usepackage{ifpdf}
\usepackage{latexsym}
\usepackage{mathtools}
\usepackage{multirow}
\usepackage{psfrag}
\usepackage{setspace}
\usepackage{subfigure}
\usepackage{tikz}
\usepackage{times}
\usepackage{tkz-berge}
\usepackage{xparse}
\usepackage{color}

\newtheorem{theorem}{Theorem}
\newtheorem{definition}{Definition}
\newtheorem{proposition}{Proposition}
\newtheorem{lemma}{Lemma}

\newtheorem{remark}{Remark}
\newtheorem{example}{Example}

\title{Deterministic and Probabilistic Conditions for Finite Completability of Low-rank Multi-View Data}

\author{Morteza Ashraphijuo, Xiaodong Wang and Vaneet Aggarwal\thanks{\noindent 
Morteza Ashraphijuo and Xiaodong Wang are with the Department of Electrical Engineering, Columbia University, NY, email: \{ashraphijuo,wangx\}@ee.columbia.edu. Vaneet Aggarwal is with the School of Industrial Engineering, Purdue University, West Lafayette, IN, email: vaneet@purdue.edu. }}
 
\begin{document}
\maketitle

\begin{abstract}

We consider the multi-view data completion problem, i.e., to complete a matrix $\mathbf{U}=[\mathbf{U}_1|\mathbf{U}_2]$ where the ranks of $\mathbf{U},\mathbf{U}_1$, and $\mathbf{U}_2$ are given. In particular, we investigate the fundamental conditions on the sampling pattern, i.e., locations of the sampled entries for finite completability of such a multi-view data given the corresponding rank constraints. In contrast with the existing analysis on Grassmannian manifold for a single-view matrix, i.e., conventional matrix completion, we propose a geometric analysis on the manifold structure for multi-view data to incorporate more than one rank constraint. We provide a deterministic necessary and sufficient condition on the sampling pattern for finite completability. We also give a probabilistic condition in terms of the number of samples per column that guarantees finite completability with high probability. Finally, using the developed tools, we derive the deterministic and probabilistic guarantees for unique completability.
\end{abstract}


\newpage
\section{Introduction}

High-dimensional data analysis has received significant recent attention due to the ubiquitous big data, including images and videos, product ranking datasets, gene expression database, etc. Many real-world high-dimensional datasets exhibit low-rank structures, i.e., the data can be represented in a much lower dimensional form \cite{donoho-high,donoho-high2}. Efficiently exploiting such low-rank structure for analyzing large high-dimensional datasets is one of the most active research area in machine learning and data mining. In this paper, we consider the multi-view low-rank data completion problem, where the ranks of the first and second views, as well as the rank of whole data consisting of both views together, are given.

The single-view learning problem has plenty of applications in various areas including signal processing \cite{Image}, network coding \cite{network}, etc.  The multi-view learning problem also finds applications in signal processing \cite{subrahmanya2010sparse}, multi-label image classification \cite{liu2015low,di2012view,zhang20n}, image retrieval \cite{tao2006asymmetric}, image synthesis \cite{shon2005learning,yu201}, data classification \cite{emailk}, multi-lingual text categorization \cite{wq}, etc.

Given a sampled matrix and a rank constraint, any matrix that agrees with the sampled entries and rank constraint is called a completion. A sampled matrix with a rank constraint is finitely completable if and only if there exist only finitely many completions of it. Most literature on matrix completion focus on developing optimization methods to obtain a completion. For single-view learning, methods including alternating minimization \cite{jain2013low,hardt2014understanding}, convex relaxation of rank \cite{candes,candes2,ashraphijuoc,cai}, etc., have been proposed. One generalization of the matrix completion problem is tensor completion where the number of orders can be more than two and alternating minimization methods \cite{liulow2,wang2016tensor} and other optimization-based methods \cite{liu2016tensor,7347424}, etc., have been proposed. Moreover, for multi-view learning, many optimization-based algorithms  have been proposed recently \cite{xu7,qxw,badch,bickeq,chaud1,oqt,qhk,dhillo,farq,ku0,mgx}.

The optimization-based matrix completion algorithms typically require incoherent conditions, which constrains the values of the entries (sampled and non-sampled) to obtain a completion with high probability. Moreover, the fundamental completability conditions that are independent of the specific completion algorithms have also been investigated. Specifically, deterministic conditions on the locations of the sampled entries (sampling pattern) are obtained through algebraic geometry analyses on Grassmannian manifold that lead to finite/unique solutions to the matrix completion problem \cite{charact,kiraly2}. In particular, in \cite{charact} a deterministic sampling pattern is proposed that is necessary and sufficient for finite completability of the sampled matrix of the given rank. Such an algorithm-independent condition can lead to a much lower sampling rate than the one that is required by the optimization-based completion algorithms. In \cite{ashraphijuo4}, \cite{ashraphijuo} and
\cite{ashraphijuo3}, we proposed a geometric analysis on canonical polyadic (CP), Tucker and tensor-train (TT) manifolds for low-rank tensor completion problem and provided the necessary and sufficient conditions on sampling pattern for finite completability of tensor given its CP, Tucker and TT ranks, respectively. However, the analysis on Grassmannian manifold in \cite{charact} is not capable of incorporating more than one rank constraint, and the analysis on Tucker manifold in \cite{ashraphijuo} is not capable of incorporating rank constraints for different views. In this paper, we investigate the finite completability problem for multi-view data by proposing an analysis on the manifold structure for such data.

Consider a sampled data matrix $\mathbf{U}$ that is partitioned as $\mathbf{U}=[\mathbf{U}_1|\mathbf{U}_2]$, where $\mathbf{U}_1$ and $\mathbf{U}_2$ are the first and second views of $\mathbf{U}$. The multi-view matrix completion problem is to complete $\mathbf{U}$ given the ranks of $\mathbf{U},\mathbf{U}_1$, and $\mathbf{U}_2$. Let $\mathbf{\Omega}$ be the sampling pattern matrix of $\mathbf{U}$, where $\mathbf{\Omega}(x,y)=1$ if $\mathbf{U}(x,y)$ is sampled and  $\mathbf{\Omega}(x,y)=0$ otherwise. This paper is mainly concerned with the following  three problems.


\begin{itemize}
	\item {\bf Problem (i)}: Characterize the necessary and sufficient conditions on $\mathbf{\Omega}$, under which there exist only finite completions of $\mathbf{U}$ that satisfy all three rank constraints. 

	\item {\bf Problem (ii)}: Characterize sufficient conditions on $\mathbf{\Omega}$, under which there exists only one completion of $\mathbf{U}$ that satisfy all three rank constraints.

	
	\item {\bf Problem (iii)}: Give lower bounds on the number of sampled entries per column such that the proposed conditions on $\mathbf{\Omega}$ for finite/unique completability are satisfied with high probability.
	
\end{itemize}	
	
 This work is inspired by \cite{charact}, where the analysis on Grassmannian manifold is proposed to solve similar problems for a single-view matrix. Specifically, in \cite{charact}  a novel approach is proposed to consider the rank factorization of a matrix and to treat each observed entry as a polynomial in terms of the entries of the components of the rank factorization. Then, under the genericity assumption, the algebraic independence among the mentioned polynomials is studied. In this paper, we consider the multi-view matrix and follow the general approach that is similar to that in \cite{charact} to treat the above problems. However, since the manifold structure for the multi-view data is fundamentally different from the Grassmannian manifold, we need to develop almost every step anew.	
	
	
	
The remainder of this paper is organized as follows. In Section \ref{backg}, the notations and problem statement are provided. In Section \ref{fincomdetsec}, we characterize the deterministic necessary and sufficient condition on the sampling pattern for finite completability. In Section \ref{probfinsec}, a probabilistic guarantee for finite completability is proposed where the condition is in terms of the number of samples per column -- in contrast with the geometric structure given in Section \ref{fincomdetsec}. In Section \ref{uniqdetprosec}, deterministic and probabilistic guarantees for unique completability are provided. Numerical results are provided in Section \ref{simusecn} to compare the number of samples per column for finite and unique completions based on our proposed analysis versus the existing method.  Finally, Section \ref{concsec} concludes the paper. 

\section{Problem Statement and A Motivating Example}\label{backg}




Let $\mathbf{U}$ be the sampled matrix to be completed. Denote $\mathbf{\Omega}$ as the sampling pattern matrix that is of the same size as $\mathbf{U}$ and $\mathbf{\Omega}(x_1,x_2)=1$ if $\mathbf{U}(x_1,x_2)$ is observed and  $\mathbf{\Omega}(x_1,x_2)=0$ otherwise. For each subset of columns $\mathbf{U}^{\prime} $ of $\mathbf{U}$, define $N_{\mathbf{\Omega}}({\mathbf{U}^{\prime}})$ as the number of observed entries in $\mathbf{U}^{\prime}$ according to the sampling pattern $\mathbf{\Omega}$. For any real number $x$, define $x^+=\max\{0,x\}$. Also, $\mathbf{I}_n$ denotes an $n \times n$ identity matrix and $\mathbf{0}_{n \times m}$ denotes an $n \times m$ all-zero matrix.


The sampled matrix $\mathbf{U} \in \mathbb{R}^{n \times (m_1+m_2)}$ is sampled. Denote a partition of $\mathbf{U}$ as $\mathbf{U}= [\mathbf{U}_1|\mathbf{U}_2]$ where $\mathbf{U}_1 \in \mathbb{R}^{n \times m_1}$ and $\mathbf{U}_2 \in \mathbb{R}^{n \times m_2}$ represent the first and second views of data, respectively. Given the rank constraints $\text{rank}(\mathbf{U}_1)=r_1$, $\text{rank}(\mathbf{U}_2)=r_2$ and $\text{rank}(\mathbf{U})=r$, we are interested in characterizing the conditions on the sampling pattern matrix $\mathbf{\Omega}$ under which there are infinite, finite, or unique completions for the sampled matrix $\mathbf{U}$. 

In \cite{charact} a necessary and sufficient condition on the sampling pattern is given for the finite completability of a matrix $\mathbf{U}$ given $\text{rank}(\mathbf{U})=r$, based on an algebraic geometry analysis on the Grassmannian manifold. However, such analysis cannot be used to treat the above multi-view problem since it is not capable of incorporating the three rank constraints simultaneously. In particular, if we obtain the conditions in \cite{charact} corresponding to  $\mathbf{U}$, $\mathbf{U}_1$ and $\mathbf{U}_2$ respectively and then take the intersections of them, it will result in a {\em sufficient} condition (not necessary) on the sampling pattern matrix $\mathbf{\Omega}$ under which there are finite number of completions of $\mathbf{U}$.

Next, we provide an example such that: (i) given $r_1$, $\mathbf{U}_1$ is infinitely completable; (ii) given $r_2$, $\mathbf{U}_2$ is infinitely completable;{\footnote{(i) and (ii) together result that given $r_1$ and $r_2$, $\mathbf{U}_1$, $\mathbf{U}_2$ and $\mathbf{U}$ are infinitely completable.}} (iii) given $r$, $\mathbf{U}$ is infinitely completable; and (iv) given $r_1$, $r_2$ and $r$, $\mathbf{U}$ is finitely completable. In other words, if $\mathcal{S}_1$ denotes the set of completions of $\mathbf{U}$ given $\text{rank}(\mathbf{U}_1)=r_1$, $\mathcal{S}_2$ denotes the set of completions of $\mathbf{U}$ given $\text{rank}(\mathbf{U}_2)=r_2$ and $\mathcal{S}$ denotes the set of completions of $\mathbf{U}$ given $\text{rank}(\mathbf{U})=r$, then in the following example $|\mathcal{S}_1|=\infty$, $|\mathcal{S}_2|=\infty$, $|\mathcal{S}|=\infty$ and $|\mathcal{S}_1 \cap \mathcal{S}_2 \cap \mathcal{S}| <\infty$.


Consider a matrix $\mathbf{U} \in \mathbb{R}^{4 \times 4}$, where $\mathbf{U}=[\mathbf{U}_1|\mathbf{U}_2]$, $\mathbf{U}_1 \in \mathbb{R}^{4 \times 2}$ (the first two columns) and $\mathbf{U}_2 \in \mathbb{R}^{4 \times 2}$ (the last two columns). Assume that $r_1 = 1$, $r_2=2$ and $r=2$. Moreover, suppose that the sampled entries of $\mathbf{U}$ are shown below. 
\[ 
\vphantom{
    \begin{matrix}{|c|c|c|c|}
    \overbrace{XYZ}^{\mbox{$R$}}\\ \\ \\ \\ \\ \\ 
    \underbrace{pqr}_{\mbox{$S$}}
    \end{matrix}}%
\begin{bmatrix}
\coolover{\mathbf{U}_1}{\times & \times} & \coolover{\mathbf{U}_2}{\times & \times}\\ 
\times &  {\color{Goldenrod}-} & \times & \times \\ 
\times & {\color{Goldenrod}-} & \times &  {\color{Goldenrod}-}  \\ 
\coolunder{\mathbf{U}}{{\color{Goldenrod}-} & {\color{Goldenrod}-} & \times & \times}  
\end{bmatrix}
\]
\vspace{-15mm}


We have the following observations about the number of completions of each matrix.

\begin{itemize}
\item {\it Given $r_1=1$, $\mathbf{U}_1$ is infinitely completable}: For any value of the $(4,1)$-th entry of $\mathbf{U}_1$, there exists exactly one completion of $\mathbf{U}_1$. Hence, there exist infinitely completions of $\mathbf{U}_1$.
\item {\it Given $r_2=2$, $\mathbf{U}_2$ is infinitely completable}: Observe that each value of the $(3,2)$-th entry of $\mathbf{U}_2$, corresponds to one completion of $\mathbf{U}_2$. As a result, there are infinitely many completions of $\mathbf{U}_2$.
\item {\it Given $r=2$, $\mathbf{U}$ is infinitely completable}: Note that for any value of the $(2,2)$-th entry of $\mathbf{U}$, there exists at least one completion of $\mathbf{U}$ (as the second column of $\mathbf{U}$ is a linear combination of two vectors and only one entry of this column is known), and therefore $\mathbf{U}$ is infinitely completable.
\item {\it For almost every matrix $\mathbf{U}$, given $r_1=1$, $r_2=2$ and $r=2$,  $\mathbf{U}$ is finitely completable}: We prove this statement in Appendix \ref{app1} by applying Theorem \ref{finitecomplthm} which takes advantage of an geometric analysis on the manifold structure for multi-view data (which is not Grassmannian manifold) to incorporate all three rank constraints simultaneously.
\end{itemize} 

\section{Finite Completability}\label{fincomdetsec}

In Section \ref{geosec}, we define an equivalence relation among all bases of the sampled matrix $\mathbf{U}$, where a basis is a set of $r$ vectors ($r=\text{rank}(\mathbf{U})$) that spans the column space of  $\mathbf{U}$. This equivalence relation leads to the manifold structure for multi-view data to incorporate all three rank constraints. We introduce a set of polynomials according to the sampled entries to analyze finite completability through analyzing algebraic independence of the defined polynomials.

In Section \ref{algebsec}, we analyze the required maximum number of algebraically independent polynomials that is necessary and sufficient for finite completability. Then, a relationship between the maximum number of algebraically independent polynomials (among the defined polynomials) and the sampling pattern (locations of the sampled entries) is characterized. This general idea is similar to the one in \cite{charact}. However, as the structure of the polynomials and the corresponding manifold are different from those of single-view matrix, we develop all original results for the multi-view framework. For example, the equivalence class and canonical structure that are developed in this paper and the structure of the polynomials require us to develop lemmas and theorems that cannot be simply obtained by generalizing the single-view results. Consequently, we obtain the necessary and sufficient condition on sampling pattern for finite completability. 

\subsection{Geometry of the Basis}\label{geosec}

Let define $r_1^{\prime} = r - r_2$, $r_2^{\prime} = r - r_1$ and $r^{\prime} = r - r_1^{\prime} - r_2^{\prime} = r_1 + r_2 -r $. Observe that $r_1 \leq r$, $r_2 \leq r$ and $r \leq r_1 + r_2$. Suppose that the basis $\mathbf{V} \in \mathbb{R}^{n \times r}$ is such that its first $r_1$ columns constitute a basis for the first view $\mathbf{U}_1$, its last $r_2$ columns constitute a basis for the second view $\mathbf{U}_2$, and all $r$ columns of $\mathbf{V}$ constitute a basis for $\mathbf{U}=[\mathbf{U}_1|\mathbf{U}_2]$, as shown in Figure \ref{fig0}. Assume that $\mathbf{V}= [\mathbf{V}_1|\mathbf{V}_2|\mathbf{V}_3]$, where $\mathbf{V}_1 \in \mathbb{R}^{n \times r_1^{\prime}}$, $\mathbf{V}_2 \in \mathbb{R}^{n \times r^{\prime}}$ and $\mathbf{V}_3 \in \mathbb{R}^{n \times r_2^{\prime}}$. Then, $[\mathbf{V}_1|\mathbf{V}_2]$ is a basis for $\mathbf{U}_1$ and $[\mathbf{V}_2|\mathbf{V}_3]$ is a basis for $\mathbf{U}_2$. As a result, there exist matrices $\mathbf{T}_1 \in \mathbb{R}^{r_1 \times m_1}$ and $\mathbf{T}_2 \in \mathbb{R}^{r_2 \times m_2}$ such that
\begin{subequations} \label{decom}
	\begin{eqnarray}
	\mathbf{U}_1 = [\mathbf{V}_1|\mathbf{V}_2] \cdot \mathbf{T}_1 , \label{decom1}\\
	\mathbf{U}_2 = [\mathbf{V}_2|\mathbf{V}_3] \cdot \mathbf{T}_2 . \label{decom2}
	\end{eqnarray}
\end{subequations}

\begin{figure}[h]
	\centering
		{\includegraphics[width=4cm]{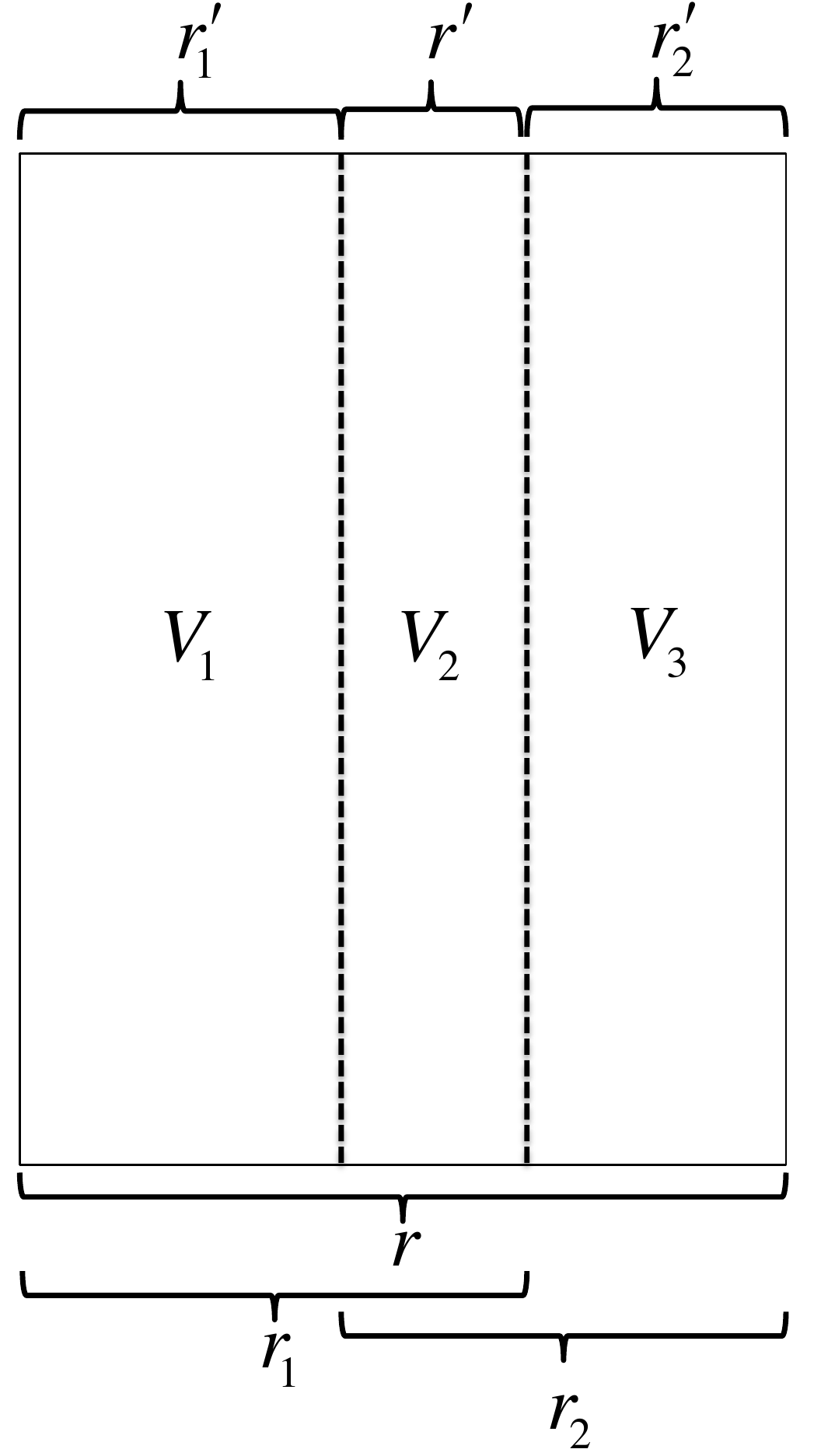}}
	\caption{ \footnotesize A basis $\mathbf{V}$ for the sampled matrix $\mathbf{U}$.}
	\label{fig0}\vspace{-4mm}
\end{figure}

For any $i_1 \in  \{1,\dots,m_1\}$ and $i_2 \in  \{1,\dots,m_2\}$, \eqref{decom} can be written as 
\begin{subequations} \label{decomcol}
	\begin{eqnarray}
	& \mathbf{U}_1\left(:,i_1\right) = [\mathbf{V}_1|\mathbf{V}_2] \cdot \mathbf{T}_1\left(:,i_1\right) , \label{decomcol1}\\
	& \mathbf{U}_2\left(:,i_2\right) = [\mathbf{V}_2|\mathbf{V}_3] \cdot \mathbf{T}_2\left(:,i_2\right) . \label{decomcol2}
	\end{eqnarray}
\end{subequations}



Let $\mathcal{M}(r,r_1,r_2,\mathbb{R}^n)$ denote the manifold structure of subspaces $\mathbf{V}$ described above for the multi-view matrix and define $\mathbb{P}_{\mathcal{M}}$ as the uniform measure on this manifold.  Moreover, define $\mathbb{P}_{{L}}$ as the Lebesgue measure on $\mathbb{R}^{r \times (m_1+m_2)}$. We assume that $\mathbf{U}$ is chosen generically from $\mathcal{M}(r,r_1,r_2,\mathbb{R}^n)$, or in other words, the entries of $\mathbf{U}$ are drawn independently with respect to Lebesgue measure on $\mathcal{M}(r,r_1,r_2,\mathbb{R}^n)$. Hence, any statement that holds for $\mathbf{U}$, it also holds for almost every (with probability one) data of the same size and rank with respect to the product measure $\mathbb{P}_{\mathcal{M}} \times \mathbb{P}_L$. Note that according to Proposition \ref{propotwo}, each multi-view data $\mathbf{U}$ can be uniquely represented in terms of a subspace $\mathbf{V} \in \mathcal{M}(r,r_1,r_2,\mathbb{R}^n)$.

In the following, we list some useful facts that are instrumental to the subsequent analysis.
 
\begin{itemize}
\item {\bf Fact $1$}: Observe that each observed entry in $\mathbf{U}_1$ results in a scalar equation from \eqref{decom1} or \eqref{decomcol1} that involves only all $r_1$ entries of the corresponding row of $[\mathbf{V}_1|\mathbf{V}_2]$ and all $r_1$ entries of the corresponding column of $\mathbf{T}_1$ in \eqref{decom1}. Similarly, each observed entry in $\mathbf{U}_2$ results in a scalar equation from \eqref{decom2} or \eqref{decomcol2} that involves only all $r_2$ entries of the corresponding row of $[\mathbf{V}_2|\mathbf{V}_3]$ and all $r_2$ entries of the corresponding column of $\mathbf{T}_2$ in \eqref{decom2}. Treating the entries of $\mathbf{V}$, $\mathbf{T}_1$ and $\mathbf{T}_2$ as variables (right-hand sides of \eqref{decom1} and \eqref{decom2}), each observed entry results in a polynomial in terms of these variables.

\item {\bf Fact $2$}: For any observed entry $\mathbf{U}_i(x_1,x_2)$, $x_1$ and $x_2$ specify the row index of $\mathbf{V}$ and the column index of $\mathbf{T}_i$, respectively, that are involved in the corresponding polynomial, $i=1,2$.



\item {\bf Fact $3$}:  It can be concluded from Bernstein's theorem \cite{Bernstein} that in a system of $n$ polynomials in $n$ variables with each consisting of a given set of monomials such that the coefficients are chosen with respect to the Lebesgue measure on $\mathcal{M}(r,r_1,r_2,\mathbb{R}^n)$, the $n$ polynomials are algebraically independent with probability one with respect to the product measure $\mathbb{P}_{\mathcal{M}} \times \mathbb{P}_L$, and therefore there exist only finitely many solutions (all given probabilities in this paper are with respect to this product measure). However, in the structure of the polynomials in our model, the set of involved monomials  are different for different set of polynomials, and therefore to ensure algebraically independency we need to have for any selected subset of the original $n$ polynomials, the number of involved variables should be more than the number of selected polynomials. Moreover, in a system of $n$ polynomials in $n-1$ variables (or less), the $n$ polynomials are algebraically dependent with probability one. Also, given that a system of $n$ polynomials in $n-1$ variables (or less) has one solution, it can be concluded that it has a unique solution with probability one. Similarly, in our model, this property should hold for any subset of the polynomials.
\end{itemize}

Given all observed entries $\{\mathbf{U}(x_1,x_2): \mathbf{\Omega}(x_1,x_2) = 1 \}$, we are interested in finding the number of possible solutions in terms of entries of $(\mathbf{V},\mathbf{T}_1,\mathbf{T}_2)$ (infinite, finite or unique) via investigating the algebraic independence among the polynomials. Throughout this paper, we make the following assumption.

{\bf Assumption $1$}: {\it Any column of $\mathbf{U}_1$ includes at least $r_1$ observed entries and any column of $\mathbf{U}_2$ includes at least $r_2$ observed entries.}

Observe that Assumption $1$ leads to a total of at least $m_1r_1+m_2r_2$ sampled entries of $\mathbf{U}$.

\begin{lemma}\label{coefmat}
Given a basis $\mathbf{V}=[\mathbf{V}_1|\mathbf{V}_2|\mathbf{V}_3]$ in \eqref{decom}, if Assumption $1$ holds, then there exists a unique solution $(\mathbf{T}_1,\mathbf{T}_2)$, with probability one. Moreover, if Assumption $1$ does not hold, then there are infinite number of solutions $(\mathbf{T}_1,\mathbf{T}_2)$, with probability one.
\end{lemma}

\begin{proof}
Note that only observed entries in the $i$-th column of $\mathbf{U}_1$ result in degree-$1$ polynomials in terms of the entries of $\mathbf{T}_1(:, i)  \in \mathbb{R}^{r_1 \times 1}$. As a result, exactly $r_1$ generically chosen degree-$1$ polynomials in terms of $r_1$ variables are needed to ensure there is a unique solution $\mathbf{T}_1$ in \eqref{decom}, with probability one. Moreover, having less than $r_1$ polynomials in terms of $r_1$ variables results in infinite solutions of $\mathbf{T}_1$ in \eqref{decom}, with probability one. The same arguments apply to $\mathbf{T}_2$.
\end{proof}

\begin{definition}
 Observe that given $\mathbf{V}$, each observed entry of $\mathbf{U}_1$ and $\mathbf{U}_2$ results in a degree-$1$ polynomial whose involved variables are the entries of the corresponding column of $\mathbf{T}_1$ and $\mathbf{T}_2$, respectively. We choose $r_1$ and $r_2$ observed entries of each column of $\mathbf{U}_1$ and $\mathbf{U}_2$, respectively, to obtain $(\mathbf{T}_1,\mathbf{T}_2)$. Let $\mathcal{P}(\mathbf{\Omega})$ denote all polynomials in terms of the entries of $\mathbf{V}$ obtained through the observed entries excluding the $m_1r_1+m_2r_2$ polynomials that were used to obtain $(\mathbf{T}_1,\mathbf{T}_2)$. Note that since $(\mathbf{T}_1,\mathbf{T}_2)$ is already solved in terms of $\mathbf{V}$, each polynomial in $\mathcal{P}(\mathbf{\Omega})$ is in terms of elements of $\mathbf{V}$. 
\end{definition}

Consider two bases $\mathbf{V}$ and $\mathbf{V}^{\prime}$ for the matrix $\mathbf{U}$ with the structure in \eqref{decom}. We say that $\mathbf{V}$ and $\mathbf{V}^{\prime}$ span the same space if and only if: (i) the spans of the first $r_1$ columns of $\mathbf{V}$ and $\mathbf{V}^{\prime}$ are the same, (ii) the spans of the last $r_2$ columns of $\mathbf{V}$ and $\mathbf{V}^{\prime}$ are the same, (iii) the spans of all columns of $\mathbf{V}$ and $\mathbf{V}^{\prime}$ are the same.

Therefore, $\mathbf{V}$ and $\mathbf{V}^{\prime}$ span the same space if and only if: (i) each column of $\mathbf{V}_1$ is a linear combination of the columns of $[\mathbf{V}_1^{\prime}|\mathbf{V}_2^{\prime}]$, (ii) each column of $\mathbf{V}_2$ is a linear combination of the columns of $\mathbf{V}_2^{\prime}$, and (iii) each column of $\mathbf{V}_3$ is a linear combination of the columns of $[\mathbf{V}_2^{\prime}|\mathbf{V}_3^{\prime}]$. The following equivalence class partitions all  possible bases such that any two bases in a class span the same space, i.e., the above-mentioned properties (i), (ii) and (iii) hold.

\begin{definition}\label{jmfnk}
Define an equivalence class for all bases $\mathbf{V} \in \mathbb{R}^{n \times r }$ of the sampled matrix $\mathbf{U}$ such that two bases $\mathbf{V}$ and $\mathbf{V}^{\prime}$ belong to the same class if there exist full rank matrices $\mathbf{A}_1 \in \mathbb{R}^{r_1 \times r_1^{\prime}}$, $\mathbf{A}_2 \in \mathbb{R}^{r^{\prime} \times r^{\prime}}$ and $\mathbf{A}_3 \in \mathbb{R}^{r_2 \times r_2^{\prime}}$ such that
\begin{subequations} \label{bas}
	\begin{eqnarray}
	\mathbf{V}_1 &=& [\mathbf{V}_1^{\prime}|\mathbf{V}_2^{\prime}]  \cdot \mathbf{A}_1, \label{bas1}\\
	\mathbf{V}_2 &=& \mathbf{V}_2^{\prime} \cdot \mathbf{A}_2, \label{bas2}\\
	\mathbf{V}_3 &=& [\mathbf{V}_2^{\prime}|\mathbf{V}_3^{\prime}] \cdot \mathbf{A}_3, \label{bas3}
	\end{eqnarray}
\end{subequations}
where $\mathbf{V}=[\mathbf{V}_1|\mathbf{V}_2|\mathbf{V}_3]$, $\mathbf{V}^{\prime}=[\mathbf{V}_1^{\prime}|\mathbf{V}_2^{\prime}|\mathbf{V}_3^{\prime}]$, $\mathbf{V}_1,\mathbf{V}_1^{\prime}  \in \mathbb{R}^{n \times r_1^{\prime}}$, $\mathbf{V}_2,\mathbf{V}_2^{\prime}  \in \mathbb{R}^{n \times r^{\prime}}$ and $\mathbf{V}_3,\mathbf{V}_3^{\prime}  \in \mathbb{R}^{n \times r_2^{\prime}}$.
\end{definition}


 Note that \eqref{bas} leads to the fact that the dimension of the space of all bases $\mathbf{V}$ in one particular class, i.e., the degree of freedom for the bases in one particular class, is equal to $nr - r_1  r_1^{\prime} -  r^{\prime}  r^{\prime} - r_2  r_2^{\prime} = nr -r^2 -r_1^2 -r_2^2 + r(r_1 +r_2) $. 

\begin{definition}\label{classp1}
{\bf (Canonical basis)} As shown in Figure \ref{fign}, denote 
\begin{subequations}\label{canonical}
\begin{eqnarray}
\mathbf{B}_1 &=& \mathbf{V}\left(1:r_1^{\prime},1:r_1^{\prime} \right) \in \mathbb{R}^{r_1^{\prime} \times r_1^{\prime}}, \label{canonical1} \\
\mathbf{B}_2 &=& \mathbf{V}\left(1:r_2^{\prime},1+r_1:r_2^{\prime}+r_1 \right) \in \mathbb{R}^{r_2^{\prime} \times r_2^{\prime}}, \label{canonical2} \\
\mathbf{B}_3 &=& \mathbf{V}\left(1+\max(r_1^{\prime},r_2^{\prime}):r^{\prime}+\max(r_1^{\prime},r_2^{\prime}),1+r_1^{\prime}:r^{\prime}+r_1^{\prime} \right) \in \mathbb{R}^{r^{\prime} \times r^{\prime}}, \label{canonical3} \\
\mathbf{B}_4 &=& \mathbf{V}\left(1+\max(r_1^{\prime},r_2^{\prime}):r^{\prime}+\max(r_1^{\prime},r_2^{\prime}),1:r_1^{\prime} \right) \in \mathbb{R}^{r^{\prime} \times r_1^{\prime}}, \label{canonical4} \\
\mathbf{B}_5 &=& \mathbf{V}\left(1+\max(r_1^{\prime},r_2^{\prime}):r^{\prime}+\max(r_1^{\prime},r_2^{\prime}), 1+r_1:r_2^{\prime}+r_1 \right) \in \mathbb{R}^{r^{\prime} \times r_2^{\prime}}. \label{canonical5}
\end{eqnarray}
\end{subequations}
Then, we call $\mathbf{V}$ a canonical basis if $\mathbf{B}_1 = \mathbf{I}_{r_1^{\prime}}$, $\mathbf{B}_2 = \mathbf{I}_{r_2^{\prime}}$, $\mathbf{B}_3 = \mathbf{I}_{r^{\prime}}$, $\mathbf{B}_4 = \mathbf{0}_{r^{\prime} \times r_1^{\prime}}$ and $\mathbf{B}_5 = \mathbf{0}_{r^{\prime} \times r_2^{\prime}}$.
\end{definition}

\begin{figure}[h]
	\centering
		{\includegraphics[width=6.5cm]{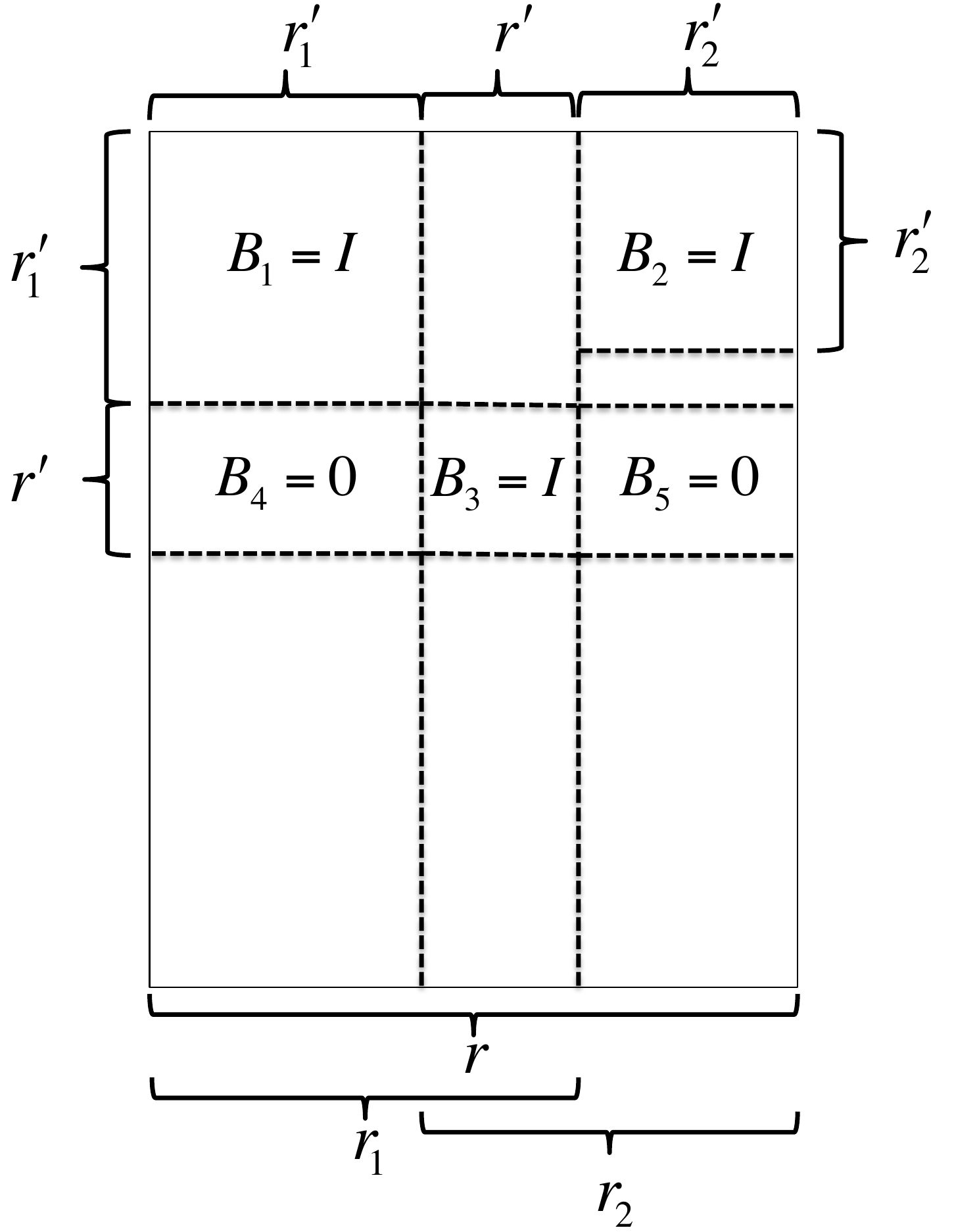}}
	\caption{ \footnotesize A canonical basis.}
	\label{fign}\vspace{-4mm}
\end{figure}

\begin{example}
Consider an example in which $\mathbf{U}=[\mathbf{U}_1|\mathbf{U}_2] \in \mathbb{R}^{4 \times 7}$ and $\mathbf{U}_1 \in \mathbb{R}^{4 \times 3}$ is the first view and $\mathbf{U}_2 \in \mathbb{R}^{4 \times 4}$ is the second view. Assume that $r_1 = 2$, $r_2=3$ and $r=4$. Then, the corresponding canonical basis is as follows.

\newcommand*{\temp}{\multicolumn{1}{c|}{$0$}}
\newcommand*{\tempp}{\multicolumn{1}{c|}{$1$}}
\newcommand*{\tempq}{\multicolumn{1}{|c}{$1$}}
\newcommand*{\temps}{\multicolumn{1}{|c}{$0$}}
\newcommand*{\tempr}{\multicolumn{1}{|c|}{$1$}}
\newcommand*{\tempt}{\multicolumn{1}{|c|}{$0$}}
$$
\mathbf{V} =\left[
\begin{array}{ cccc } \cline{1-1} \cline{3-4}
\tempr & {\color{Goldenrod}-} & \tempq & \temp   \\ \cline{1-1}
{\color{Goldenrod}-} & {\color{Goldenrod}-} & \temps & \tempp  \\ \cline{1-4}
\tempt & $1$ & \temps & \temp  \\ \cline{1-4}
{\color{Goldenrod}-} & {\color{Goldenrod}-} & {\color{Goldenrod}-} &  {\color{Goldenrod}-}
\end{array}\right]
$$

Observe that $r^2+r_1^2+r_2^2-r(r_1+r_2)=9$ of the entries are known.
\end{example}

 The following proposition shows the uniqueness of the canonical basis for the single-view matrix.

\begin{proposition}\label{propoone}
Assume that $\mathbf{X} \in \mathbb{R}^{n_1 \times n_2}$ is generically chosen from the manifold of $n_1 \times n_2$ matrices of rank $r$. For almost every $\mathbf{X}$, there exists a unique basis $\mathbf{Y} \in \mathbb{R}^{n_1 \times r}$ for $\mathbf{X}$ such that $\mathbf{Y}(1:r,1:r) = \mathbf{I}_r$, where $\mathbf{Y}$ is a basis for $\mathbf{X}$ if each column of $\mathbf{X}$ can be written as a linear combination of the columns of $\mathbf{Y}$ and $\mathbf{Y}(1:r,1:r)$ represents the submatrix of $\mathbf{Y}$ that consists of the first $r$ columns and the first $r$ rows and $\mathbf{I}_r$ denotes the $r \times r$ identity matrix.
\end{proposition}

\begin{proof}
Consider an arbitrary basis $\mathbf{Y}^{\prime} \in \mathbb{R}^{n_1 \times r}$ for $\mathbf{X}$. Let $\mathcal{V}$ denote the set of all full rank $n_1 \times r$ matrices whose column span is equal to the column span of $\mathbf{Y}^{\prime}$ and note that $\mathcal{V}$ is the set of all bases with $r$ columns for $\mathbf{X}$. Consider an arbitrary member of set $\mathcal{V}$ and denote it by $\mathbf{Y}^{\prime \prime}$. Since, the column span of $\mathbf{Y}^{\prime}$ and the column span of $\mathbf{Y}^{\prime \prime}$ are the same, each column of $\mathbf{Y}^{\prime \prime}$ can be written as a linear combination of columns of $\mathbf{Y}^{\prime}$, and therefore there exists a unique full rank $\mathbf{Z} \in \mathbb{R}^{r \times r}$ such that $\mathbf{Y}^{\prime \prime} = \mathbf{Y}^{\prime} \mathbf{Z}$. Note that if $\mathbf{Z}$ is not full rank, we conclude $\mathbf{Y}^{\prime \prime}$ is not full rank as well, which contradicts the assumption.

Moreover, genericity of $\mathbf{X}$ results that each $r \times r$ submatrix of $\mathbf{Y}^{\prime}$ is full rank, with probability one. This is because we have $\mathbf{X} = \mathbf{Y}^{\prime} \mathbf{T}$ for some $\mathbf{T} \in \mathbb{R}^{r \times n_2}$, and therefore the fact that the submatrix consisting of any $r$ rows of $\mathbf{X}$ is full rank results that the submatrix consisting of any $r$ rows of $\mathbf{Y}^{\prime}$ is full rank as well. Let $\mathbf{Z}_0$ denote the inverse of $\mathbf{Y}^{\prime}(1:r,1:r)$, i.e., $\mathbf{Z}_0 = \left( \mathbf{Y}^{\prime}(1:r,1:r) \right)^{-1}$. Therefore, $\mathbf{Y}= \mathbf{Y}^{\prime} \mathbf{Z}_0$ is the unique basis for $\mathbf{X}$ that $\mathbf{Y}(1:r,1:r) = \mathbf{I}_r$.
\end{proof}

The following proposition considers the multi-view data described in this paper and shows the uniqueness of the canonical basis.

\begin{proposition}\label{propotwo}
For almost every $\mathbf{U}$, there exists a unique basis $\mathbf{V} \in \mathbb{R}^{n \times r}$ for $\mathbf{U}$ such that $\mathbf{V}$ satisfies the canonical pattern in Definition \ref{classp1} and also its first $r_1$ columns constitute a basis for the first view $\mathbf{U}_1$, its last $r_2$ columns constitute a basis for the second view $\mathbf{U}_2$, and all $r$ columns of $\mathbf{V}$ constitute a basis for $\mathbf{U}=[\mathbf{U}_1|\mathbf{U}_2]$.
\end{proposition}

\begin{proof}
Consider an arbitrary basis $\mathbf{V}^{\prime}=[\mathbf{V}_1^{\prime}|\mathbf{V}_2^{\prime}|\mathbf{V}_3^{\prime}]$ for $\mathbf{U}$ such that its first $r_1$ columns constitute a basis for the first view $\mathbf{U}_1$, its last $r_2$ columns constitute a basis for the second view $\mathbf{U}_2$, and all $r$ columns of $\mathbf{V}$ constitute a basis for $\mathbf{U}=[\mathbf{U}_1|\mathbf{U}_2]$. Let $\mathcal{V}$ denote the set of all such bases for $\mathbf{U}$ and consider an arbitrary member of this set and denote it by $\mathbf{V}=[\mathbf{V}_1|\mathbf{V}_2|\mathbf{V}_3]$. Hence, according to the earlier discussion before Definition \ref{jmfnk}, equations \eqref{bas1}-\eqref{bas3} hold. This is because the column spans of the first $r_1$ columns of $\mathbf{V}$ and $\mathbf{V}^{\prime}$, or the column spans of the last $r_2$ columns of $\mathbf{V}$ and $\mathbf{V}^{\prime}$ and also the column spans of the all $r$ columns of $\mathbf{V}$ and $\mathbf{V}^{\prime}$ are the same.

Similar to the proof of  Proposition \ref{propoone},  \eqref{bas2} results that the unique $\mathbf{V}_2$ that can satisfy the pattern $\mathbf{B}_3 = \mathbf{I}_{r^{\prime}}$ in Definition \ref{classp1}, can be obtained by $\mathbf{A}_2$ equal to the inverse of matrix \\ $\mathbf{V}^{\prime} \left(1+\max(r_1^{\prime},r_2^{\prime}):r^{\prime}+\max(r_1^{\prime},r_2^{\prime}),1+r_1^{\prime}:r^{\prime}+r_1^{\prime} \right)$. In order to complete the proof, it suffices to show that there exists a unique $\mathbf{A}_1$ that results in satisfying the patterns $\mathbf{B}_1 = \mathbf{I}_{r_1^{\prime}}$ and $\mathbf{B}_4 = \mathbf{0}_{r^{\prime} \times r_1^{\prime}}$ in Definition \ref{classp1} for $\mathbf{V}_1$ (the uniqueness of $\mathbf{A}_3$ is similar to that for $\mathbf{A}_1$).

Let $\mathbf{A}_1^{\prime} \in \mathbb{R}^{r_1 \times r_1}$ denote the inverse of $r_1 \times r_1$ matrix $[\mathbf{V}_1|\mathbf{V}_2](1:r_1,1:r_1)$. Note that existing of the inverse is a consequence of genericity assumption which results the submatrix consisting of any $r_1$ rows of $\mathbf{U}_1$ is full rank with probability one, and therefore the submatrix consisting of any $r_1$ rows of $[\mathbf{V}_1|\mathbf{V}_2]$ is full rank. Let $\mathbf{A}_1$ be the first $r_1^{\prime}$ columns of $\mathbf{A}_1^{\prime}$. Then, $\mathbf{A}_1$ ensures that the patterns $\mathbf{B}_1 = \mathbf{I}_{r_1^{\prime}}$ and $\mathbf{B}_4 = \mathbf{0}_{r^{\prime} \times r_1^{\prime}}$ in Definition \ref{classp1} hold for $\mathbf{V}_1$. Finally, note that $\mathbf{A}_1$ is unique. Otherwise, there exist two different inverse matrices for the full rank $r_1 \times r_1$ matrix $[\mathbf{V}_1|\mathbf{V}_2](1:r_1,1:r_1)$, which is contradiction.
\end{proof}

\begin{remark}\label{class}
 In order to prove there are finitely many completions for the matrix $\mathbf{U}$, it suffices to prove that given $\mathbf{T}_1$ and $\mathbf{T}_2$, there are finitely many canonical bases that fit in $\mathbf{U}$, where a basis fitting in $\mathbf{U}$ is equivalent to the existence of a completion of $\mathbf{U}$ such that each of its columns can be written as a linear combination of the corresponding basis.
\end{remark}

Note that patterns $\mathbf{B}_1$ and $\mathbf{B}_4$ are in $\mathbf{V}_1$, patterns $\mathbf{B}_2$ and $\mathbf{B}_5$ are in $\mathbf{V}_3$, and pattern $\mathbf{B}_3$ is in $\mathbf{V}_2$. It can be easily seen that according to the definition of the equivalence class in \eqref{bas}, any permutation of the rows of any of these patterns satisfies the property that in each class there exists exactly one basis with the permuted pattern.

\subsection{Algebraic Independence}\label{algebsec}

The following theorem provides a condition on the polynomials in $\mathcal{P}(\mathbf{\Omega})$ that is equivalent (necessary and sufficient) to finite completability of $\mathbf{U}$.

\begin{theorem}\label{thmdimbas}
Assume that Assumption $1$ holds. For almost every sampled matrix $\mathbf{U}$, there are at most finitely many bases that fit in $\mathbf{U}$ if and only if there exist $ nr -r^2 -r_1^2 -r_2^2 + r(r_1 +r_2) $ algebraically independent polynomials in $\mathcal{P}(\mathbf{\Omega})$.
\end{theorem}

\begin{proof}
According to Lemma \ref{coefmat}, Assumption $1$ results that $(\mathbf{T}_1,\mathbf{T}_2)$ can be determined uniquely (finitely). Let $\mathcal{P}(\mathbf{\Omega}) = \{p_1,\dots,p_t\}$ and define $\mathcal{S}_i$ as the set of all basis $\mathbf{V}$ that satisfy polynomials $\{p_1,\dots,p_i\}$, $i=0,1,\dots,t$, where $\mathcal{S}_0$ is the set of all bases $\mathbf{V}$ without any polynomial restriction. Each polynomial in terms of the entries of $\mathbf{V}$ reduces the degree of freedom or the dimension of the set of solutions by one. Therefore, $\text{dim}(\mathcal{S}_i)=\text{dim}(\mathcal{S}_{i-1})$ if the maximum number of algebraically independent polynomials in sets $\{p_1,\dots,p_i\}$ and $\{p_1,\dots,p_{i-1}\}$ are the same and $\text{dim}(\mathcal{S}_i)=\text{dim}(\mathcal{S}_{i-1})-1$ otherwise.

Observe that the number of variables is $\text{dim}(\mathcal{S}_0)=nr -r^2 -r_1^2 -r_2^2 + r(r_1 +r_2)$ and the number of solutions of the system of polynomials $\mathcal{P}(\mathbf{\Omega})$ is $|\mathcal{S}_t|$. Therefore, using Fact $3$, with probability one $|\mathcal{S}_t|$ is a finite number if and only if $\text{dim}(\mathcal{S}_t)=0$. As mentioned earlier, the dimension of the set of all bases without any polynomial restriction, i.e., $\text{dim}(\mathcal{S}_0) =  nr -r^2 -r_1^2 -r_2^2 + r(r_1 +r_2) $. Hence, we conclude that the existence of exactly $ nr -r^2 -r_1^2 -r_2^2 + r(r_1 +r_2) $ algebraically independent polynomials in $\mathcal{P}(\mathbf{\Omega})$ is equivalent to having finitely many bases, i.e., finite completability of $\mathbf{U}$ with probability one.
\end{proof}

In this subsection, we are interested in characterizing a relationship between the sampling pattern $\mathbf{\Omega}$ and the maximum number of algebraically independent polynomials in $\mathcal{P}(\mathbf{\Omega})$. To this end, we construct a {\bf constraint matrix} $\mathbf{\breve{\Omega}}$ based on $\mathbf{\Omega}$ such that each column of $\mathbf{\breve{\Omega}}$ represents exactly one of the polynomials in $\mathcal{P}(\mathbf{\Omega})$.

Consider an arbitrary column of the first view $\mathbf{U}_1\left(:,i\right)$, where $i \in  \{1,\dots,m_1\}$. Let $l_i = N_{\mathbf{\Omega}} (\mathbf{U}_1\left(:,i\right))$ denote the number of observed entries in the $i$-th column of the first view. Assumption $1$ results that $l_i \geq r_1$. 

We construct $l_i - r_1$ columns with binary entries based on the locations of the observed entries in $\mathbf{U}_1\left(:,i\right)$ such that each column has exactly $r_1+1$ entries equal to one. Assume that $x_1, \dots, x_{l_i}$ be the row indices of all observed entries in this column. Let $\mathbf{\Omega}^{i}_1$ be the corresponding $n \times (l_i - r_1)$ matrix to this column which is defined as the following: for any $j \in \{1,\dots,l_i-r_1\}$, the $j$-th column has the value $1$ in rows $\{x_1,\dots , x_{r_1},x_{r_1+j}\}$ and zeros elsewhere. Define the binary constraint matrix of the first view as $\mathbf{\breve{\Omega}}_1 = \left[\mathbf{\Omega}^{1}_1|\mathbf{\Omega}^{2}_1\dots|\mathbf{\Omega}^{m_1}_{1} \right] \in \mathbb{R}^{n \times K_1}$ \cite{charact}, where $K_1 = N_{\mathbf{\Omega}} (\mathbf{U}_1)-m_1r_1$.

Similarly, we construct the binary constraint matrix $\mathbf{\breve{\Omega}}_2 \in \mathbb{R}^{n \times K_2}$ for the second view, where $K_2 = N_{\mathbf{\Omega}} (\mathbf{U}_2)-m_2r_2$. Define the constraint matrix of $\mathbf{U}$ as $\mathbf{\breve{\Omega}}=[\mathbf{\breve{\Omega}}_1 | \mathbf{\breve{\Omega}}_2] \in \mathbb{R}^{n \times (K_1+K_2)}$. For any subset of columns $\mathbf{\breve{\Omega}}^{\prime}$ of $\mathbf{\breve{\Omega}}$, $\mathcal{P}(\mathbf{\breve{\Omega}}^{\prime})$ denotes the subset of $\mathcal{P}(\mathbf{{\Omega}})$ that corrseponds to $\mathbf{\breve{\Omega}}^{\prime}$.{\footnote{Note that only $r_1$ and $r_2$ appear in the structure of constraint matrix and $r$ only appears in the structure of the basis as the basis has $r$ columns.}} In this paper, when we refer to a subset of columns of the constraint matrix, those columns are assumed to correspond to different columns of $\Omega$.


\begin{example}
Consider the same example as in Section \ref{backg}, where matrix $\mathbf{U}=[\mathbf{U}_1|\mathbf{U}_2] \in \mathbb{R}^{4 \times 4}$ and $\mathbf{U}_1 \in \mathbb{R}^{4 \times 2}$ is the first view and $\mathbf{U}_2 \in \mathbb{R}^{4 \times 2}$ is the second view. The samples that are used to obtain $(\mathbf{T}_1, \mathbf{T}_2)$ are colored as red in the following.  Assume that $r_1 = 1$, $r_2=2$ and $r=2$. Then, the constraint matrix is as the following.

$$
\mathbf{U} =\left[
\begin{array}{ cccc }
\coolover{\mathbf{U}_1}{{\color{red}\times} & {\color{red}\times}} & \coolover{\mathbf{U}_2}{{\color{red}\times} & {\color{red}\times}}\\ 
{\times} & {\color{Goldenrod}-} & {\color{red}\times} & {\color{red}\times}  \\
{\times} & {\color{Goldenrod}-} & {\times} &  {\color{Goldenrod}-} \\
{\color{Goldenrod}-} & {\color{Goldenrod}-} & {\times} &  {\times}
\end{array}\right]
\Longrightarrow
\mathbf{\breve{\Omega}} =
\left[
\begin{array}{ ccccc }
\coolover{\mathbf{\breve{\Omega}}_1}{1 & 1} & \coolover{\mathbf{\breve{\Omega}}_2}{1 & 1 & 1}\\ 
1 & 0 & 1 & 1 & 1 \\ 
0 & 1 & 1 & 0 & 0 \\ 
0 & 0 & 0 & 1 & 1 
\end{array}
\right].
$$

\end{example}


Assume that $\mathbf{\breve{\Omega}}^{\prime}$ is an arbitrary subset of columns of the constraint matrix $\mathbf{\breve{\Omega}}$. Then, $\mathbf{\breve{\Omega}}^{\prime}_1$ and $\mathbf{\breve{\Omega}}^{\prime}_2$ denote the columns that correspond to the first and second views, respectively. Similarly, assume that $\mathbf{\Omega}^{\prime}$ is an arbitrary subset of columns of $\mathbf{\Omega}$. Then, $\mathbf{\Omega}^{\prime}_1$ and $\mathbf{\Omega}^{\prime}_2$ denote the columns that correspond to the first view and second view, respectively. Moreover, for any matrix $\mathbf{X}$, $c(\mathbf{X})$ denotes the number of columns of $\mathbf{X}$ and $g(\mathbf{X})$ denotes the number of nonzero rows of $\mathbf{X}$.

The following lemma gives an upper bound on the maximum number of algebraically independent polynomials in any subset of columns of the constraint matrix $\mathbf{\breve{\Omega}}$. Simply put, for a set of polynomials with coefficients chosen generically, the total number of involved variables in the polynomials is an upper bound for the maximum number of algebraically independent polynomials.

\begin{lemma}\label{maxnumbpol}
Assume that Assumption $1$ holds. Let $\mathbf{\breve{\Omega}}^{\prime}$ be an arbitrary subset of columns of the constraint matrix $\mathbf{\breve{\Omega}}$. Then, the maximum number of algebraically independent polynomials in $\mathcal{P}(\mathbf{\breve{\Omega}}^{\prime})$ is upper bounded by 
\begin{eqnarray}\label{maxindepchar}
r_1^{\prime} (g(\mathbf{\breve{\Omega}}^{\prime}_1)-r_1)^+ + r_2^{\prime} (g(\mathbf{\breve{\Omega}}^{\prime}_2)-r_2)^+ + r^{\prime} (g(\mathbf{\breve{\Omega}}^{\prime})-r^{\prime})^+.
\end{eqnarray}
\end{lemma}

\begin{proof}
Using Fact $3$, the maximum number of algebraically independent polynomials in $\mathcal{P}(\mathbf{\breve{\Omega}}^{\prime})$ is at most equal to the number of involved variables in the polynomials. Note that each observed entry of $\mathbf{U}_1$ results in a polynomial that involves all $r_1$ entries of a row of $\mathbf{V}_1$. As a result, the number of entries of $\mathbf{V}_1$ that are involved in the polynomials is exactly $r_1 g(\mathbf{\breve{\Omega}}^{\prime}_1)$. As mentioned earlier, the rows of patterns $\mathbf{B}_1$ and $\mathbf{B}_4$ in Definition \ref{classp1} can be permuted such that exactly one basis in each class satisfies the new pattern. Hence, it can be permuted such that $r_1$ rows of $\mathbf{B}_1$ and $\mathbf{B}_4$ are a subset of the nonzero rows of $\mathbf{\breve{\Omega}}^{\prime}_1$ since there are at least $r_1 +1$ nonzero rows (any column of the constraint matrix of the first view includes exactly $r_1+1$ nonzero entries). Recall that the total number of known entries of $\mathbf{V}_1$ is the summation of the number of entries of $\mathbf{B}_1$ and $\mathbf{B}_4$, i.e., $r_1^{\prime} r_1$. Therefore, the number of variables (unknown entries) of $\mathbf{V}_1$ that are involved in $\mathcal{P}(\mathbf{\breve{\Omega}}^{\prime})$ is equal to $r_1^{\prime} (g(\mathbf{\breve{\Omega}}^{\prime}_1)-r_1)^+$. Note that $g(\mathbf{\breve{\Omega}}^{\prime}_1)-r_1$ is negative if and only if $g(\mathbf{\breve{\Omega}}^{\prime}_1)=0$.

Similarly, the number of unknown entries of $\mathbf{V}_2$ and $\mathbf{V}_3$ that are invloved in $\mathcal{P}(\mathbf{\breve{\Omega}}^{\prime})$ are $r^{\prime} (g(\mathbf{\breve{\Omega}}^{\prime})-r^{\prime})^+$ and $r_2^{\prime} (g(\mathbf{\breve{\Omega}}^{\prime}_2)-r_2)^+$, respectively. Therefore, the number of unknown entries of basis $\mathbf{V}$ that are involved in $\mathcal{P}(\mathbf{\breve{\Omega}}^{\prime})$ is equal to $r_1^{\prime} (g(\mathbf{\breve{\Omega}}^{\prime}_1)-r_1)^+ + r_2^{\prime} (g(\mathbf{\breve{\Omega}}^{\prime}_2)-r_2)^+ + r^{\prime} (g(\mathbf{\breve{\Omega}}^{\prime})-r^{\prime})^+$.
\end{proof}

 A set of polynomials is called minimally algebraically dependent if the polynomials in that set are algebraically dependent but the polynomials in any of its proper subsets are algebraically independent. The next lemma which is Lemma $3$ in \cite{ashraphijuo}, states an important property of a set of minimally algebraically dependent among polynomials in $\mathcal{P}(\mathbf{\breve{\Omega}})$. This lemma is needed to derive the the maximum number of algebraically independent polynomials in any subset of $\mathcal{P}(\mathbf{\breve{\Omega}})$. Note that $c(\mathbf{\breve{\Omega}}^{\prime})$ is the number of polynomials in $\mathcal{P}(\mathbf{\breve{\Omega}}^{\prime})$.

\begin{lemma}\label{minimaldeppolnew}
Assume that Assumption $1$ holds.{\footnote{Assumption $1$ is only needed to construct $\mathbf{\breve{\Omega}}$ and not in the proof. Similarly, Assumption $A_j$ in \cite{ashraphijuo} is only needed to construct $\mathbf{\breve{\Omega}}$ and not in the proof. Also, note that the number of polynomials in \cite{ashraphijuo} is denoted by the size of $(j+1)^{\text{th}}$ dimension of $\mathbf{\breve{\Omega}}^{\prime}$.}} Let $\mathbf{\breve{\Omega}}^{\prime}$ be an arbitrary subset of columns of the constraint matrix $\mathbf{\breve{\Omega}}$. Assume that polynomials in $\mathcal{P}(\mathbf{\breve{\Omega}}^{\prime})$ are minimally algebraically dependent. Then, the number of variables (unknown entries) of $\mathbf{V}$ that are involved in $\mathcal{P}(\mathbf{\breve{\Omega}}^{\prime})$ is equal to $c(\mathbf{\breve{\Omega}}^{\prime})-1$.
\end{lemma}

The following lemma explicitly characterizes the relationship between the number of algebraically independent polynomials in $ \mathcal{P}(\mathbf{\breve{\Omega}})$ and the geometry of $\mathbf{\breve{\Omega}}$. 

\begin{lemma}\label{indepdeppoly}
Assume that Assumption $1$ holds. Let $\mathbf{\breve{\Omega}}^{\prime}$ be an arbitrary subset of columns of the constraint matrix $\mathbf{\breve{\Omega}}$. The polynomials in $\mathcal{P}(\mathbf{\breve{\Omega}}^{\prime})$ are algebraically dependent if and only if there exists $\mathbf{\breve{\Omega}}^{\prime \prime } \subseteq \mathbf{\breve{\Omega}}^{\prime}$ such that
\begin{eqnarray}\label{depcondi}
r_1^{\prime} (g(\mathbf{\breve{\Omega}}^{\prime \prime }_1)-r_1)^+ + r_2^{\prime} (g(\mathbf{\breve{\Omega}}^{\prime \prime }_2)-r_2)^+ + r^{\prime} (g(\mathbf{\breve{\Omega}}^{\prime \prime })-r^{\prime})^+ < c(\mathbf{\breve{\Omega}}^{\prime \prime }).
\end{eqnarray}
\end{lemma}

\begin{proof}
Assume that there exists $\mathbf{\breve{\Omega}}^{\prime \prime } \subseteq \mathbf{\breve{\Omega}}^{\prime}$ such that \eqref{depcondi} holds. Note that there are $c(\mathbf{\breve{\Omega}}^{\prime \prime })$ polynomials in the set $\mathcal{P}(\mathbf{\breve{\Omega}}^{\prime \prime})$. Hence, according to Lemma \ref{maxnumbpol} and \eqref{depcondi}, the maximum number of algebraically independent polynomials is less than the number of polynomials, i.e., $c(\mathbf{\breve{\Omega}}^{\prime \prime })$. Therefore, the polynomials in $\mathcal{P}(\mathbf{\breve{\Omega}}^{\prime \prime})$ and therefore the polynomials in $\mathcal{P}(\mathbf{\breve{\Omega}}^{\prime})$ are algebraically dependent.

For the converse, suppose that the polynomials in $\mathcal{P}(\mathbf{\breve{\Omega}}^{\prime})$ are algebraically dependent. Hence, there exists a subset of these polynomials, $\mathcal{P}(\mathbf{\breve{\Omega}}^{\prime \prime})$, such that the polynomials are minimally algebraically dependent. According to Lemma \ref{minimaldeppolnew}, the number of variables involved in the polynomials of $\mathcal{P}(\mathbf{\breve{\Omega}}^{\prime \prime})$ is $c(\mathbf{\breve{\Omega}}^{\prime \prime})-1$. 

On the other hand, as mentioned in the proof of Lemma \ref{maxnumbpol}, the minimum number of involved variables (unknown entries of $\mathbf{V}$) is equal to $r_1^{\prime} (g(\mathbf{\breve{\Omega}}^{\prime \prime}_1)-r_1)^+ + r_2^{\prime} (g(\mathbf{\breve{\Omega}}^{\prime \prime}_2)-r_2)^+ + r^{\prime} (g(\mathbf{\breve{\Omega}}^{\prime \prime})-r^{\prime})^+$, which is therefore less than or equal to $c(\mathbf{\breve{\Omega}}^{\prime \prime})-1$ and the proof is complete. 
\end{proof}

Finally, the next theorem which is the main result of this subsection gives the necessary and sufficient condition on $\mathbf{\breve{\Omega}}$ to ensure there exist $nr -r^2 -r_1^2 -r_2^2 + r(r_1 +r_2)$ algebraically independent polynomials in $\mathcal{P}(\mathbf{\Omega})$, and therefore it gives the necessary and sufficient condition on $\mathbf{\breve{\Omega}}$ for finite completability of $\mathbf{U}$. 

\begin{theorem}\label{finitecomplthm}
Assume that Assumption $1$ holds. For almost every $\mathbf{U}$, the sampled matrix $\mathbf{U}$ is finite completable if and only if there exists a subset of columns $\mathbf{\breve{\Omega}}^{\prime} \in \mathbb{R}^{n \times m}$ of the constraint matrix $\mathbf{\breve{\Omega}}$ such that $m= nr -r^2 -r_1^2 -r_2^2 + r(r_1 +r_2)$ and for any subset of columns $\mathbf{\breve{\Omega}}^{\prime \prime}$ of $\mathbf{\breve{\Omega}}^{\prime}$ the following inequality holds
\begin{eqnarray}\label{charfinitecom}
r_1^{\prime} (g(\mathbf{\breve{\Omega}}^{\prime \prime }_1)-r_1)^+ + r_2^{\prime} (g(\mathbf{\breve{\Omega}}^{\prime \prime }_2)-r_2)^+ + r^{\prime} (g(\mathbf{\breve{\Omega}}^{\prime \prime })-r^{\prime})^+ \geq c(\mathbf{\breve{\Omega}}^{\prime \prime }).
\end{eqnarray}
\end{theorem}

\begin{proof}
According to Theorem \ref{thmdimbas}, with probability one, the sampled matrix $\mathbf{U}$ is finitely completable if and only if there exist $ nr -r^2 -r_1^2 -r_2^2 + r(r_1 +r_2)$ algebraically independent polynomials in $\mathcal{P}(\mathbf{\breve{\Omega}})$. On the other hand, according to Lemma \ref{indepdeppoly}, there exist $ nr -r^2 -r_1^2 -r_2^2 + r(r_1 +r_2)$ algebraically independent polynomials in $\mathcal{P}(\mathbf{\breve{\Omega}})$ if and only if there exists a subset of columns $\mathbf{\breve{\Omega}}^{\prime}$ with $ nr -r^2 -r_1^2 -r_2^2 + r(r_1 +r_2)$ columns of the constraint matrix $\mathbf{\breve{\Omega}}$ that satisfies \eqref{charfinitecom} for any of its subset of columns.
%
\end{proof}

\section{Probabilistic Guarantees for Finite Completability}\label{probfinsec}

In this section, we show that if the number of samples in each column satisfies a proposed lower bound, then the conditions stated in the statement of Theorem \ref{finitecomplthm} on sampling pattern hold, i.e., $\mathbf{U}$ is finitely completable with high probability.

The next lemma will be used to prove Theorem \ref{thmfincomprob}. More specifically, in Theorem \ref{thmfincomprob} we consider three disjoint sets of columns of $\mathbf{U}$ and apply Lemma \ref{fincomsampprob} to each of them. Then, we combine the three sets of columns and show that they satisfy the conditions stated in the statement of Theorem \ref{finitecomplthm}.

\begin{lemma}\label{fincomsampprob}
Assume that $r^{\prime \prime} \leq \frac{n}{6}$ and also each column of $\mathbf{\Omega}$ includes at least $l$ nonzero entries, where 
\begin{eqnarray}\label{minl1}
l > \max\left\{9 \ \log \left( \frac{n}{\epsilon} \right) + 3 \ \log \left( \frac{k}{\epsilon} \right) + 6, 2r^{\prime \prime}\right\}. 
\end{eqnarray}
Let $\mathbf{\Omega}^{\prime}$ be an arbitrary set of $n -r^{\prime \prime}$ columns of $\mathbf{\Omega}$. Then, with probability at least $1-\frac{\epsilon}{k}$, every subset $\mathbf{\Omega}^{\prime \prime}$ of columns of $\mathbf{\Omega}^{\prime}$ satisfies 
\begin{eqnarray}\label{proper1}
g(\mathbf{\Omega}^{\prime \prime}) - r^{\prime \prime} \geq c(\mathbf{\Omega}^{\prime \prime}).
\end{eqnarray}
\end{lemma}

\begin{proof}
Please refer to the proof of \cite[Lemma $9$]{charact}. Note that the only difference is that the last inequalities of $(16)$ and $(18)$ in \cite{charact} should now be upper bounded by $\frac{\epsilon}{rd}$ instead of $\frac{\epsilon}{d^2}$.
\end{proof}

\begin{theorem}\label{thmfincomprob}
Assume that the following inequalities hold 
\begin{eqnarray}
\frac{n}{6} &\geq & \max\{r_1,r_2,r^{\prime}\}, \label{assum1} \\ 
m_{1} &\geq & r_1^{\prime}(n-r_1), \label{assum2} \\
m_{2} &\geq & r_2^{\prime}(n-r_2), \label{assum3} \\
m_1+m_2 &\geq & r_1^{\prime}(n-r_1) + r_2^{\prime}(n-r_2) + r^{\prime}(n-r^{\prime}). \label{assum4}
\end{eqnarray}
Moreover, assume that each column of $\mathbf{\Omega}$ includes at least $l$ nonzero entries, where 
\begin{eqnarray}\label{minl2}
l > \max\left\{9 \ \log \left( \frac{n}{\epsilon} \right) + 3 \ \max\left\{\log \left( \frac{3r_1^{\prime}}{\epsilon}\right), \log \left( \frac{3r_2^{\prime}}{\epsilon}\right), \log \left( \frac{3r^{\prime}}{\epsilon} \right)\right\} + 6, 2r_1, 2r_2\right\}. 
\end{eqnarray}
Then, with probability at least $1-\epsilon$, $\mathbf{U}$ is finitely completable.
\end{theorem}

\begin{proof}
Let $\mathbf{\Omega}_1^{\prime}$ be an arbitrary set of $n-r_1$ columns of $\mathbf{\Omega}_1$. Note that having \eqref{minl2}, it is easy to see that \eqref{minl1} holds with $k$ and $r^{\prime \prime}$ replaced by $3r_1^{\prime}$ and $r_1$, respectively. Hence, having \eqref{assum1}, Lemma \ref{fincomsampprob} results that any subset of columns $\mathbf{\Omega}_1^{\prime \prime}$ of $\mathbf{\Omega}_1^{\prime}$ satisfies
\begin{eqnarray}\label{proper2}
g(\mathbf{\Omega}_1^{\prime \prime}) - r_1 \geq c(\mathbf{\Omega}_1^{\prime \prime}),
\end{eqnarray}
with probability at least $1-\frac{\epsilon}{3r_1^{\prime}}$. According to Lemma \ref{graphgen} below and by setting $r = r_1$, as a subset of columns $\mathbf{\Omega}_1^{\prime}$ of $\mathbf{\Omega}_1$ satisfies \eqref{proper2}, there exists a subset of columns $\mathbf{\breve{\Omega}}_1^{\prime}$ of the constraint matrix of the first view $\mathbf{\breve{\Omega}}_1$ (corresponding columns to the columns of $\mathbf{\Omega}_1^{\prime}$) that satisfies \eqref{proper2} as well.

Assumption \eqref{assum2} results that $\mathbf{\Omega}_1$ includes at least $r_1^{\prime}(n-r_1)$ columns or in other words, $r_1^{\prime}$ disjoint sets of columns each including $n-r_1$ columns. All  $r_1^{\prime}$ disjoint sets satisfy property (i) simultaneously with probability at least $1-\frac{\epsilon}{3}$. Therefore, there exist $r_1^{\prime}$ disjoint sets of columns each including $n-r_1$ columns of the constraint matrix of the first view $\mathbf{\breve{\Omega}}_1$, and also all $r_1^{\prime}$ disjoint sets satisfy \eqref{proper2}, simultaneously with probability at least $1-\frac{\epsilon}{3}$. Let $\mathbf{\breve{\bar{\Omega}}}_1$ denote the union of the $r_1^{\prime}$ mentioned sets of columns.
 
Consider any subset of columns $\mathbf{\breve{\bar{\Omega}}}_1^{\prime}$ of $\mathbf{\breve{\bar{\Omega}}}_1$ and define $\mathbf{\breve{\bar{\Omega}}}_{1,i}^{\prime}$ as the intersection of $\mathbf{\breve{\bar{\Omega}}}_1^{\prime}$ and the $i$-th set among the mentioned $r_1^{\prime}$ sets for $i=1,\dots,r_1^{\prime}$. Without loss of generality, assume that $\max_{1 \leq i \leq r_1^{\prime}} \{c(\mathbf{\breve{\bar{\Omega}}}_{1,i}^{\prime})\} = c(\mathbf{\breve{\bar{\Omega}}}_{1,1}^{\prime})$. Then, 
\begin{eqnarray}
c(\mathbf{\breve{\bar{\Omega}}}_1^{\prime}) = \sum_{i=1}^{r_1^{\prime}} c(\mathbf{\breve{\bar{\Omega}}}_{1,i}^{\prime}) \leq r_1^{\prime} c(\mathbf{\breve{\bar{\Omega}}}_{1,i}^{\prime}) \leq r_1^{\prime} (g(\mathbf{\breve{\bar{\Omega}}}_{1,1}^{\prime}) - r_1)^+ \leq r_1^{\prime} (g(\mathbf{\breve{\bar{\Omega}}}_1^{\prime}) - r_1)^+,
\end{eqnarray}
where the second inequality follows from \eqref{proper2}. Therefore, we have
\begin{eqnarray}\label{property1}
c(\mathbf{\breve{\bar{\Omega}}}_1^{\prime}) \leq r_1^{\prime} (g(\mathbf{\breve{\bar{\Omega}}}_1^{\prime}) - r_1)^+.
\end{eqnarray}

Note that having \eqref{minl2}, it is easy to see that \eqref{minl1} holds with $k$ and $r^{\prime \prime}$ replaced by $3r_2^{\prime}$ and $r_2$, respectively. Moreover, recall that $r^{\prime} = r_1 + r_2 -r \leq \min\{r_1,r_2\}$, and therefore, having \eqref{minl2}, it is easy to see that \eqref{minl1} holds with $k$ and $r^{\prime \prime}$ replaced by $3r^{\prime}$ and $r^{\prime}$, respectively. As a result, similarly, having \eqref{assum1} and \eqref{assum3}, $\mathbf{\breve{\Omega}}_2$ includes $r_2^{\prime}(n-r_2)$ columns $\mathbf{\breve{\bar{\Omega}}}_2$ that with probability at least $1-\frac{\epsilon}{3}$ for any subset of it $\mathbf{\breve{\bar{\Omega}}}_2^{\prime}$ we have
\begin{eqnarray}\label{property2}
c(\mathbf{\breve{\bar{\Omega}}}_2^{\prime}) \leq  r_2^{\prime} (g(\mathbf{\breve{\bar{\Omega}}}_2^{\prime}) - r_2)^+.
\end{eqnarray}

Using \eqref{assum4}, $\mathbf{\Omega}$ includes $r^{\prime}(n-r^{\prime})$ columns $\mathbf{\bar{\Omega}}$ (disjoint from $\mathbf{\bar{\Omega}}_1$ and $\mathbf{\bar{\Omega}}_2$ corresponding to $\mathbf{\breve{\bar{\Omega}}}_1$ and $\mathbf{\breve{\bar{\Omega}}}_2$). Similar to $\mathbf{\breve{\bar{\Omega}}}_1$ and $\mathbf{\breve{\bar{\Omega}}}_2$, $\mathbf{\breve{{\Omega}}}$ includes $r^{\prime}(n-r^{\prime})$ columns $\mathbf{\breve{\bar{\Omega}}}$ (disjoint from $\mathbf{\breve{\bar{\Omega}}}_1$ and $\mathbf{\breve{\bar{\Omega}}}_2$) that with probability at least $1-\frac{\epsilon}{3}$ for any subset of columns of it $\mathbf{\breve{\bar{\Omega}}}^{\prime}$ we have
\begin{eqnarray}\label{property3}
c(\mathbf{\breve{\bar{\Omega}}}^{\prime}) \leq  r^{\prime} (g(\mathbf{\breve{\bar{\Omega}}}^{\prime}) - r^{\prime})^+.
\end{eqnarray}


Therefore, any subset of columns of $\mathbf{\breve{\bar{\Omega}}}_1$ satisfies \eqref{property1} and any subset of $\mathbf{\breve{\bar{\Omega}}}_2$ satisfies \eqref{property2} and any subset of $\mathbf{\breve{\bar{\Omega}}}$ satisfies \eqref{property3} simultaneously with probability at least $1-\epsilon$. Define $\mathbf{{\breve{\Omega}}}^{\prime} = [\mathbf{\breve{\bar{\Omega}}}_1|\mathbf{\breve{\bar{\Omega}}}_2|\mathbf{\breve{\bar{\Omega}}}] \in \mathbb{R}^{n \times m}$, where 
\begin{eqnarray}
m = r^{\prime}(n-r^{\prime})+r_1^{\prime}(n-r_1)+r_2^{\prime}(n-r_2)=nr -r^2 -r_1^2 -r_2^2 + r(r_1 +r_2).
\end{eqnarray}
Let $\mathbf{{\breve{\Omega}}}^{\prime \prime}$ be a subset of columns of $\mathbf{{\breve{\Omega}}}^{\prime}$ and define $\mathbf{{\breve{\Omega}}}^{\prime \prime}_1$, $\mathbf{{\breve{\Omega}}}^{\prime \prime}_2$ and $\mathbf{{\breve{\Omega}}}^{\prime \prime}_3$ as the intersection of $\mathbf{{\breve{\Omega}}}^{\prime}$ with $\mathbf{\breve{\bar{\Omega}}}_1$, $\mathbf{\breve{\bar{\Omega}}}_2$ and $\mathbf{\breve{\bar{\Omega}}}$, respectively. Consequently, with probability at least $1-\epsilon$
\begin{eqnarray}
c(\mathbf{{\breve{\Omega}}}^{\prime \prime}) = \sum_{i=1}^{3} c(\mathbf{{\breve{\Omega}}}^{\prime \prime}_i) \leq r_1^{\prime} (g(\mathbf{{\breve{\Omega}}}^{\prime \prime}_1)-r_1)^+ + r_2^{\prime} (g(\mathbf{{\breve{\Omega}}}^{\prime \prime}_2)-r_2)^+ + r^{\prime} (g(\mathbf{{\breve{\Omega}}}^{\prime \prime}_3)-r^{\prime})^+,
\end{eqnarray}
and therefore according to Theorem \ref{finitecomplthm}, $\mathbf{U}$ is finite completable with probability at least $1-\epsilon$.
\end{proof}

The following lemma is taken from \cite[Lemma $8$]{ashraphijuo}.

\begin{lemma}\label{graphgen}
Let $R$ be a given nonnegative integer. Assume that there exists a matrix $\mathbf{\Omega}^{\prime}$ such that it consists of $n-R$ columns of $\mathbf{\Omega}$ and each column of $\mathbf{\Omega}^{\prime}$ includes at least $R+1$ nonzero entries and satisfies the following property:
\begin{itemize}
\item Denote an arbitrary matrix obtained by choosing any subset of the columns of $\mathbf{\Omega}^{\prime}$ by $\mathbf{\Omega}^{\prime \prime}$. Then, 
\begin{eqnarray}\label{proper233}
g(\mathbf{\Omega}^{\prime \prime}) -R  \geq c(\mathbf{\Omega}^{\prime \prime}).
\end{eqnarray}
\end{itemize}
Then, there exists a matrix $\mathbf{\breve{\Omega}}^{\prime}$ with the same size as $\mathbf{\Omega}^{\prime}$ such that: each column has exactly $R+1$ entries equal to one, and if $\mathbf{\breve{\Omega}}^{\prime}(x,y)=1$ then we have $\mathbf{\Omega}^{\prime}(x,y)=1$. Moreover, $\mathbf{\breve{\Omega}}^{\prime}$ satisfies the above-mentioned property.
\end{lemma}

The following lemma is taken from \cite{ashraphijuo} and is used in Lemma \ref{thmfincomprobsamppro} to find a condition on the sampling probability that results \eqref{minl2}.

\begin{lemma}\label{azuma}
Consider a vector with $n$ entries where each entry is observed with  probability  $p$  independently from the other entries. If $p > p^{\prime} = \frac{k}{n} + \frac{1}{\sqrt[4]{n}}$, then with probability  at least $\left(1-\exp(-\frac{\sqrt{n}}{2})\right)$, more than $k$ entries are observed.
\end{lemma}

\begin{lemma}\label{thmfincomprobsamppro}
Assume that the inequalities \eqref{assum1}--\eqref{assum4} hold. Moreover, assume that each entry of $\mathbf{U}$ is independently observed with probability $p$, where
\begin{eqnarray}\label{minl2prime}
p > \frac{1}{n}\max\left\{9 \ \log \left( \frac{n}{\epsilon} \right) + 3 \ \max\left\{\log \left( \frac{3r_1^{\prime}}{\epsilon}\right), \log \left( \frac{3r_2^{\prime}}{\epsilon}\right), \log \left( \frac{3r^{\prime}}{\epsilon} \right)\right\} + 6, 2r_1, 2r_2\right\} + \frac{1}{\sqrt[4]{n}}. 
\end{eqnarray}
Then, with probability at least $(1-\epsilon)\left(1-\exp(-\frac{\sqrt{n}}{2})\right)^{m_1+m_2}$, $\mathbf{U}$ is finitely completable.
\end{lemma}

\begin{proof}
Note that according to lemma \ref{azuma}, the number of observed entries of each of the $m_1+m_2$ columns satisfies \eqref{minl2} with probability at least $\left(1-\exp(-\frac{\sqrt{n}}{2})\right)$. Hence, the proof is straight-forward using Theorem \ref{thmfincomprob}.
\end{proof}


\section{Deterministic and Probabilistic Guarantees for Unique Completability}\label{uniqdetprosec}

Theorem \ref{finitecomplthm} gives the necessary and sufficient condition on sampling pattern for finite completability. Hence, even one sample short of the condition in Theorem \ref{finitecomplthm} results in infinite number of completions with probability one. However, as we showed in an example in \cite{ashraphijuo}, finite completability can be different from unique completability. We show that adding a mild condition to the conditions obtained in the analysis for Problem (i) leads to unique completability. To this end, we obtain multiple sets of minimally algebraically dependent polynomials and show that the variables involved in these polynomials can be determined uniquely, and therefore entries of $\mathbf{U}$ can be determined uniquely. 

 Recall that there exists at least one completion of $\mathbf{U}$ since the original  multi-view matrix that is sampled satisfies the rank constraints. The following lemma is a re-statement of Lemma $9$ in \cite{ashraphijuo}.

\begin{lemma}\label{minimaldeppol}
Assume that Assumption $1$ holds. Let $\mathbf{\breve{\Omega}}^{\prime}$ be an arbitrary subset of columns of the constraint matrix $\mathbf{\breve{\Omega}}$. Assume that polynomials in $\mathcal{P}(\mathbf{\breve{\Omega}}^{\prime})$ are minimally algebraically dependent. Then, all variables (unknown entries) of $\mathbf{V}$ that are involved in $\mathcal{P}(\mathbf{\breve{\Omega}}^{\prime})$ can be determined uniquely.
\end{lemma}

Theorem \ref{thmuniq} below gives a sufficient conditions on sampling pattern for unique completability. To be more specific, condition (i) in the statement of Theorem \ref{thmuniq}, i.e., $nr -r^2 -r_1^2 -r_2^2 + r(r_1 +r_2)$ algebraically independent polynomials in terms of the entries of $\mathbf{V}$, results in finite completability. Hence, adding any single polynomial to them results in a set of algebraically dependent polynomials and using Lemma \ref{minimaldeppol} some of the entries of basis $\mathbf{V}$ can be determined uniquely. Then, conditions (ii) and (iii) result in more polynomials such that all entries of $\mathbf{V}$ can be determined uniquely.



\begin{theorem}\label{thmuniq}
Suppose that Assumption $1$ holds. Moreover assume that there exist disjoint subsets of columns $\mathbf{\breve{\Omega}}^{\prime} \in \mathbb{R}^{n \times m}$, $\mathbf{\breve{\Omega}}^{\prime}_1 \in \mathbb{R}^{n \times m^{\prime}}$ and $\mathbf{\breve{\Omega}}^{\prime}_2 \in \mathbb{R}^{n \times m^{\prime \prime}}$ of the constraint matrix $\mathbf{\breve{\Omega}}$ such that the following properties hold

(i) $m= nr -r^2 -r_1^2 -r_2^2 + r(r_1 +r_2)$ and for any subset of columns $\mathbf{\breve{\Omega}}^{\prime \prime}$ of the matrix $\mathbf{\breve{\Omega}}^{\prime}$, \eqref{charfinitecom} holds.

(ii) $\mathbf{\breve{\Omega}}^{\prime}_1$ is a subset of columns of $\mathbf{\breve{\Omega}}_1$ (constraint matrix of the first view), $m^{\prime}= n- r_1$ and for any subset of columns $\mathbf{\breve{\Omega}}^{\prime \prime}_1$ of the matrix $\mathbf{\breve{\Omega}}^{\prime}_1$
\begin{eqnarray}\label{conuni1}
g(\mathbf{\breve{\Omega}}^{\prime \prime}_1) - r_1 \geq c(\mathbf{\breve{\Omega}}^{\prime \prime}_1).
\end{eqnarray}

(iii) $\mathbf{\breve{\Omega}}^{\prime}_2$ is a subset of columns of $\mathbf{\breve{\Omega}}_2$ (constraint matrix of the first view), $m^{\prime \prime}= n- r_2$ and for any subset of columns $\mathbf{\breve{\Omega}}^{\prime \prime}_2$ of the matrix $\mathbf{\breve{\Omega}}^{\prime}_2$
\begin{eqnarray}\label{conuni2}
g(\mathbf{\breve{\Omega}}^{\prime \prime}_2) - r_2 \geq c(\mathbf{\breve{\Omega}}^{\prime \prime}_2).
\end{eqnarray}

Then, with probability one, there exists exactly one completion of $\mathbf{U}$ that satisfies the rank constraints.
\end{theorem}

\begin{proof}
According to Theorem \ref{finitecomplthm}, property (i) results that there are only finitely many completions of $\mathbf{U}$ that satisfy the rank constraints. We show that having properties (ii) and (iii) results in obtaining all entries of the basis uniquely, and therefore there exists only one completion of $\mathbf{U}$. According to Theorem \ref{finitecomplthm}, the $nr -r^2 -r_1^2 -r_2^2 + r(r_1 +r_2)$ polynomials in $\mathcal{P}(\mathbf{\breve{\Omega}}^{\prime})$ are algebraically independent. As a result, by adding any single polynomial to this set, we will have a set of algebraically dependent polynomials. 

Consider a single polynomial from $\mathcal{P}(\mathbf{\breve{\Omega}}^{\prime}_1) \cup \mathcal{P}(\mathbf{\breve{\Omega}}^{\prime}_2)$ and denote it by $p_0$. Hence, polynomials in set $p_0 \cup \mathcal{P}(\mathbf{\breve{\Omega}}^{\prime})$ are algebraically dependent, and therefore there exists $\mathcal{P}^{\prime} (p_0) \subseteq \{p_0 \cup \mathcal{P}(\mathbf{\breve{\Omega}}^{\prime})\}$ such that $p_0 \in \mathcal{P}^{\prime}(p_0)$ and polynomials in $\mathcal{P}^{\prime}(p_0)$ are minimally algebraically dependent. Lemma \ref{minimaldeppol} results that all variables involved in polynomials in $\mathcal{P}^{\prime}(p_0)$ can be determined uniquely. The number entries of $\mathbf{V}$ that are involved in $\mathcal{P}^{\prime}(p_0)$ is at least $r_1$ if $p_0 \in \mathcal{P}(\mathbf{\breve{\Omega}}^{\prime}_1)$ and $r_2$ if $p_0 \in \mathcal{P}(\mathbf{\breve{\Omega}}^{\prime}_2)$. This is because the number of entries of $\mathbf{V}$ that are involved in polynomials in $\mathcal{P}^{\prime}(p_0)$ is at least equal to the number of entries of $\mathbf{V}$ that are involved in $p_0$. Hence, $\mathcal{P}^{\prime}(p_0)$ results in $r_1$ or $r_2$ polynomials that each has a unique solution.

Similarly, consider any other polynomial $p_1$ in $\mathcal{P}(\mathbf{\breve{\Omega}}^{\prime}_1) \cup \mathcal{P}(\mathbf{\breve{\Omega}}^{\prime}_2)$ and note that polynomials in set $p_1 \cup \mathcal{P}(\mathbf{\breve{\Omega}}^{\prime})$ are algebraically dependent. Hence, we can repeat the above procedure for $p_0$ for polynomial $p_1$. Repeating this procedure for any subset of polynomials in $\mathcal{P}(\mathbf{\breve{\Omega}}^{\prime \prime}_1) \cup \mathcal{P}(\mathbf{\breve{\Omega}}^{\prime \prime}_2) \subseteq \mathcal{P}(\mathbf{\breve{\Omega}}^{\prime}_1) \cup \mathcal{P}(\mathbf{\breve{\Omega}}^{\prime}_2)$ results in $r_1^{\prime} (g(\mathbf{{\breve{\Omega}}}^{\prime \prime}_1)-r_1)^+ + r_2^{\prime} (g(\mathbf{{\breve{\Omega}}}^{\prime \prime}_2)-r_2)^+ + r^{\prime} (g(\mathbf{{\breve{\Omega}}}^{\prime \prime}_3)-r^{\prime})^+$ polynomials (as this is the number of unknown entries involved in the polynomials $\mathcal{P}(\mathbf{\breve{\Omega}}^{\prime}_1) \cup \mathcal{P}(\mathbf{\breve{\Omega}}^{\prime}_2)$) and observe that \eqref{conuni1} and \eqref{conuni2} result that the number of involved unknown entries of basis is not less than the number of polynomials, and therefore they are independent. Moreover, observe that $\mathbf{\breve{\Omega}}^{\prime}_1$ and $\mathbf{\breve{\Omega}}^{\prime}_2$ are such that polynomials obtained via this procedure cover all entries of basis. Therefore, all entries of basis can be determined uniquely with probability one.
\end{proof}

The next theorem gives a probabilistic guarantee for satisfying the conditions in the statement of Theorem \ref{thmuniq} or in other words, a probabilistic guarantee for unique completability. However, similar to Theorem \ref{thmfincomprob}, the condition on sampling pattern is in terms of the number of samples per column instead of the complicated conditions in the statement of Theorem \ref{thmuniq} on the structure of sampling pattern.

\begin{theorem}\label{thmunicomprob}
Assume that the following inequalities hold 
\begin{eqnarray}
\frac{n}{6} &\geq & \max\{r_1,r_2,r^{\prime}\}, \label{assumuni1} \\ 
m_{1} &\geq & (r_1^{\prime}+1)(n-r_1), \label{assumuni2} \\
m_{2} &\geq & (r_2^{\prime}+1)(n-r_2), \label{assumuni3} \\
m_1+m_2 &\geq & (r_1^{\prime}+1)(n-r_1) + (r_2^{\prime}+1)(n-r_2) + r^{\prime}(n-r^{\prime}). \label{assumuni4}
\end{eqnarray}
Moreover, assume that each column of $\mathbf{\Omega}$ includes at least $l$ nonzero entries, where 
\begin{eqnarray}\label{minluni2}
l > \max\left\{9 \ \log \left( \frac{n}{\epsilon} \right) + 3 \ \max\left\{\log \left( \frac{6r_1^{\prime}}{\epsilon}\right), \log \left( \frac{6r_2^{\prime}}{\epsilon}\right), \log \left( \frac{6r^{\prime}}{\epsilon} \right)\right\} + 6, 2r_1, 2r_2\right\}. 
\end{eqnarray}
Then, with probability at least $1-\epsilon$, there exists exactly one completion of $\mathbf{U}$.
\end{theorem}

\begin{proof}
According to the proof of Theorem \ref{thmfincomprob}, \eqref{minluni2} results that there exists a subset of columns $\mathbf{\breve{\Omega}}^{\prime} \in \mathbb{R}^{n \times m}$ of the constraint matrix $\mathbf{\breve{\Omega}}$ such that condition (i) in the statement of Theorem \ref{thmuniq} is satisfied, with probability at least $1-\frac{\epsilon}{2}$. Then, assumptions \eqref{assumuni2}, \eqref{assumuni3} and \eqref{assumuni4} result that there exist $n-r_1$ columns $\mathbf{\breve{\Omega}^{\prime}}_1$ of $\mathbf{\breve{\Omega}}_1$ and $n-r_2$ columns $\mathbf{\breve{\Omega}^{\prime}}_2$ of $\mathbf{\breve{\Omega}}_2$ that are disjoint from $\mathbf{\breve{\Omega}}^{\prime}$. This is easily verified by comparing assumptions \eqref{assumuni2}, \eqref{assumuni3} and assumptions \eqref{assum2}, \eqref{assum3} in Theorem \ref{thmfincomprob}. 

Note that according to Lemma \ref{fincomsampprob}, \eqref{minluni2} results that $\mathbf{\breve{\Omega}}^{\prime}_1$ satisfies condition (ii) in the statement of Theorem \ref{thmuniq} with probability at least $1-\frac{\epsilon}{6}$. Similarly, \eqref{minluni2} results that $\mathbf{\breve{\Omega}}^{\prime}_2$ satisfies condition (iii) in the statement of Theorem \ref{thmuniq} with probability at least $1-\frac{\epsilon}{6}$. Therefore, all conditions in the statement of Theorem \ref{thmuniq} are satisfied simultaneously with probability at least $1 - \frac{\epsilon}{2} - \frac{\epsilon}{6} - \frac{\epsilon}{6}$. Hence, according to Theorem \ref{thmuniq}, there exists only one completion of $\mathbf{U}$ with probability at least $1-\epsilon$.
\end{proof}

\begin{remark}
Comparing assumptions \eqref{assum2}--\eqref{assum4} for finite completability with assumptions \eqref{assumuni2}--\eqref{assumuni4} for unique completability, we see there is a mild change, i.e., $r_i$ for finiteness is replaced by $r_i +1$ for uniqueness.

Moreover, the lower bound on the number of samples per column increases mildly from \eqref{minl2} for finiteness to \eqref{minluni2} for uniqueness, i.e., the factor $3$ in the log terms in \eqref{minl2} become $6$ in \eqref{minluni2}.
\end{remark}

\begin{lemma}\label{thmunicomprobsamppro}
Assume that the inequalities \eqref{assumuni1}--\eqref{assumuni4} hold. Moreover, assume that each entry of $\mathbf{U}$ is independently observed with probability $p$, where
\begin{eqnarray}\label{minl2primeuni}
p > \frac{1}{n}\max\left\{9 \ \log \left( \frac{n}{\epsilon} \right) + 3 \ \max\left\{\log \left( \frac{6r_1^{\prime}}{\epsilon}\right), \log \left( \frac{6r_2^{\prime}}{\epsilon}\right), \log \left( \frac{6r^{\prime}}{\epsilon} \right)\right\} + 6, 2r_1, 2r_2\right\} + \frac{1}{\sqrt[4]{n}}. 
\end{eqnarray}
Then, with probability at least $(1-\epsilon)\left(1-\exp(-\frac{\sqrt{n}}{2})\right)^{m_1+m_2}$, $\mathbf{U}$ is uniquely completable.
\end{lemma}

\begin{proof}
Note that according to lemma \ref{azuma}, the number of observed entries of each of the $m_1+m_2$ columns satisfies \eqref{minluni2} with probability at least $\left(1-\exp(-\frac{\sqrt{n}}{2})\right)$. Hence, the proof is straight-forward using Theorem \ref{thmunicomprob}.
\end{proof}

\section{Numerical Comparisons}\label{simusecn}

Here we compare the lower bound on the number of samples per column obtained by the proposed analysis in this paper with the bound obtained by the method in \cite{charact}. Recall that the existing method on Grassmannian manifold in \cite{charact} provides a bound on the number of samples for finite completability for a matrix $\mathbf{U}$ given $\text{rank}(\mathbf{U})=r$. As explained before this analysis cannot be used to treat the multi-view problem since it is not capable of incorporating the three rank constraints simultaneously. However, if we obtain the bound in \cite{charact} corresponding to  $\mathbf{U}$, $\mathbf{U}_1$ and $\mathbf{U}_2$ respectively and then take the maximum of them, it results in the following bound on the number of samples for finite completability 
\begin{eqnarray}\label{minl2org}
l > \max\left\{12 \ \log \left( \frac{n}{\epsilon} \right), 2r_1, 2r_2, 2r\right\}. 
\end{eqnarray}

We consider a sampled data $\mathbf{U}=[\mathbf{U}_1|\mathbf{U}_2] \in \mathbb{R}^{500 \times 100000}$, where $\mathbf{U}_1, \mathbf{U}_2 \in \mathbb{R}^{500 \times 50000}$, i.e., $n=500$ and $m_1=m_2=50000$. In Figure \ref{fignewcomp} we plot the bounds given in \eqref{minl2}  for finite completability and compare it with the one in \eqref{minl2org}, as a function of the value $r_1=r_2$, for $r=40$, $r=60$ and $r=100$, with $\epsilon = 0.0001$. Recall that $r_1, r_2 \leq r$ and $r \leq r_1+r_2$. It is seen that our proposed method requires less number of samples per column compared with the method in \cite{charact}. Note that given the large number of columns, i.e., $m=m_1+m_2=10^5$, this leads to significantly less amount  of sampled data. 

Note that the curves are not continuous as we need to apply the ceiling operator to the non-integer numbers in \eqref{minl2} and \eqref{minl2org}. Moreover, note that as both bounds in \eqref{minl2} and \eqref{minl2org} are equal to the maximum of two terms: (i) one is on the order of $\log (n)$ or $\log (n) + \log (r)$, and (ii) one is linear in $r$. Hence, by increasing the value of $r$, eventually it will be a linear function of $r$, as seen in Figure \ref{fig:subfigS4.3}. However, within most applications $r$ is typically small.


\begin{figure}[htbp]
\centering
\subfigure[$r=40$.]{
	\includegraphics[width=5.5cm]{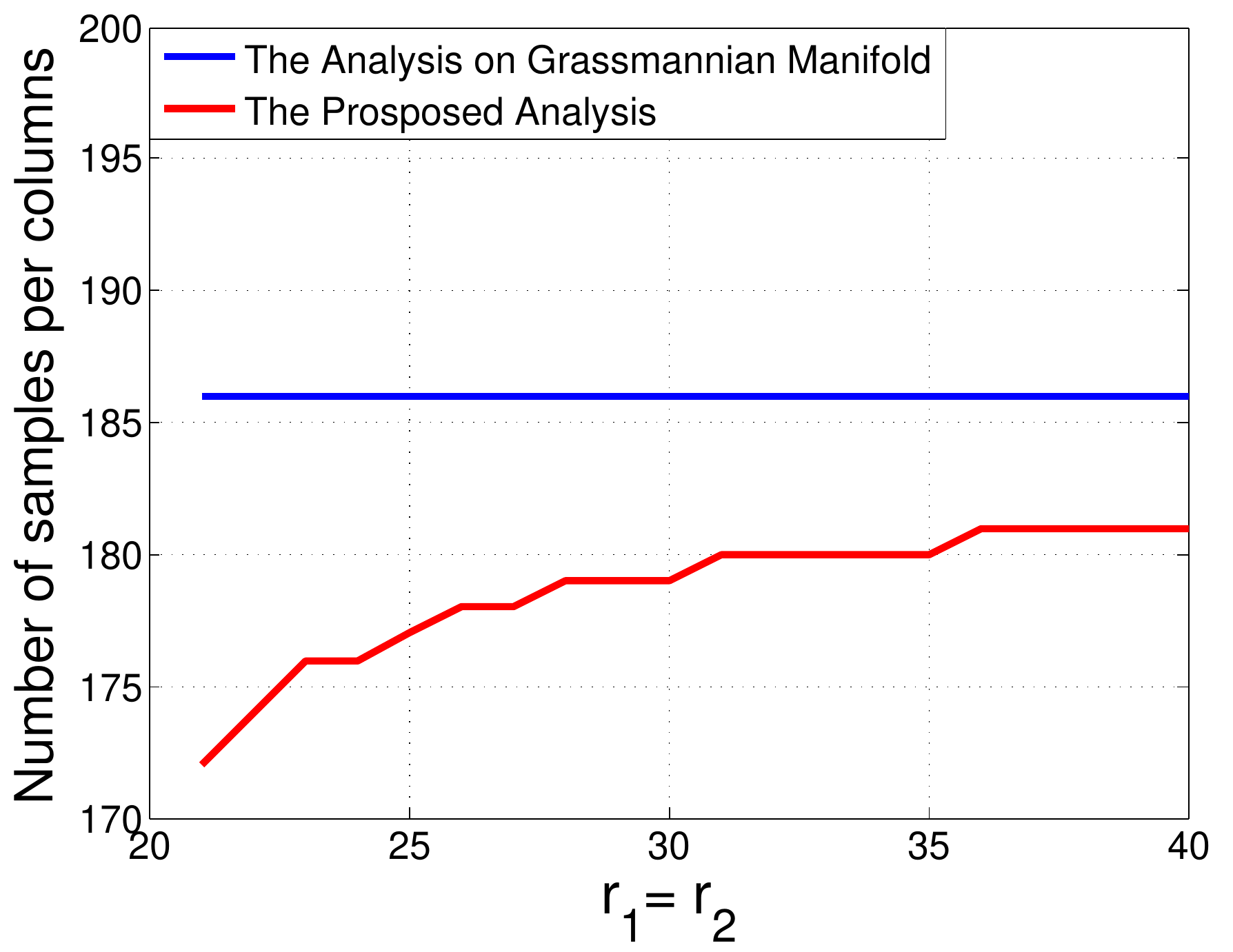}
\label{fig:subfigS4.1}
}
\subfigure[$r=60$.]{
	\includegraphics[width=5.5cm]{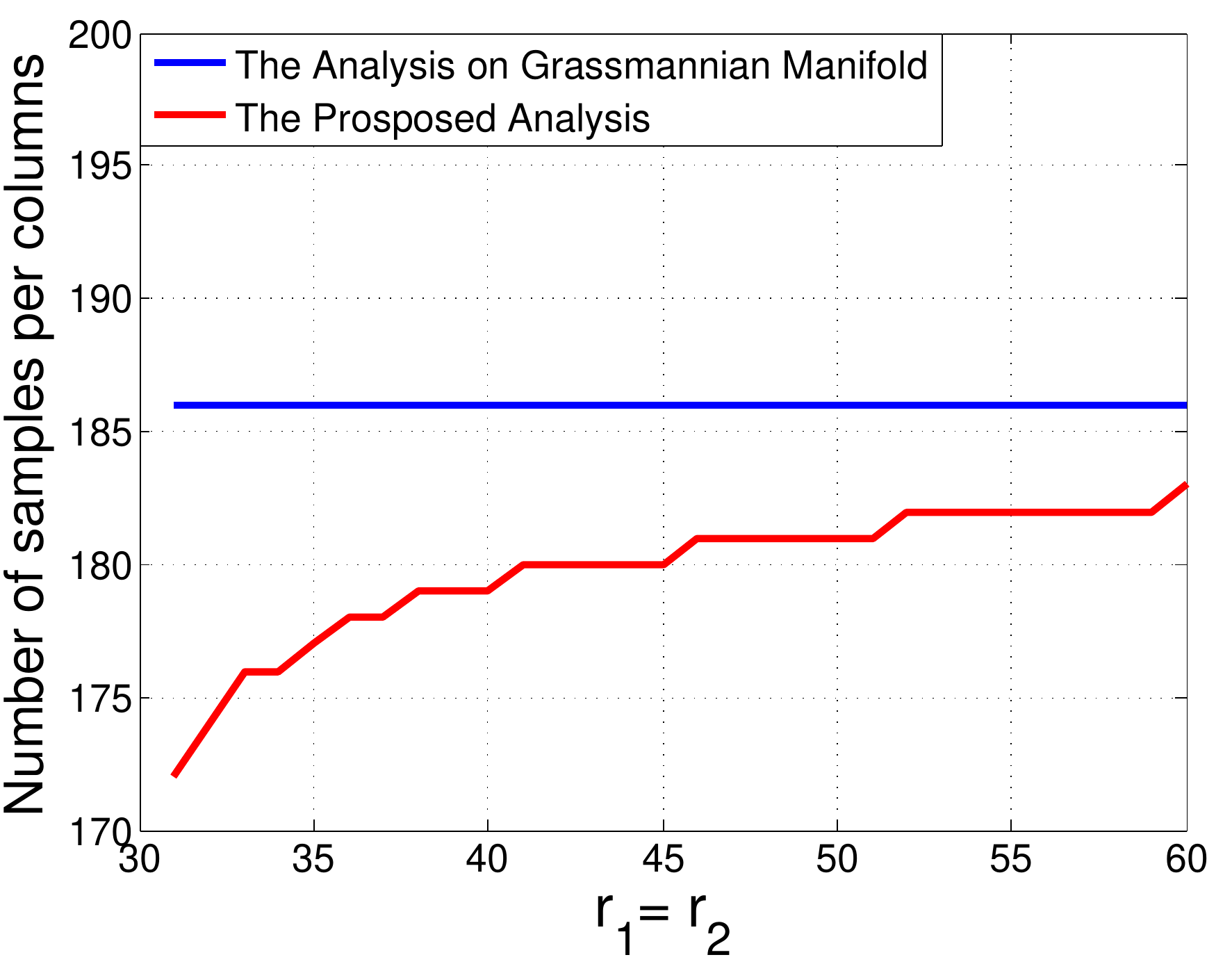}
\label{fig:subfigS4.2}
}
\subfigure[$r=100$.]{
	\includegraphics[width=5.5cm]{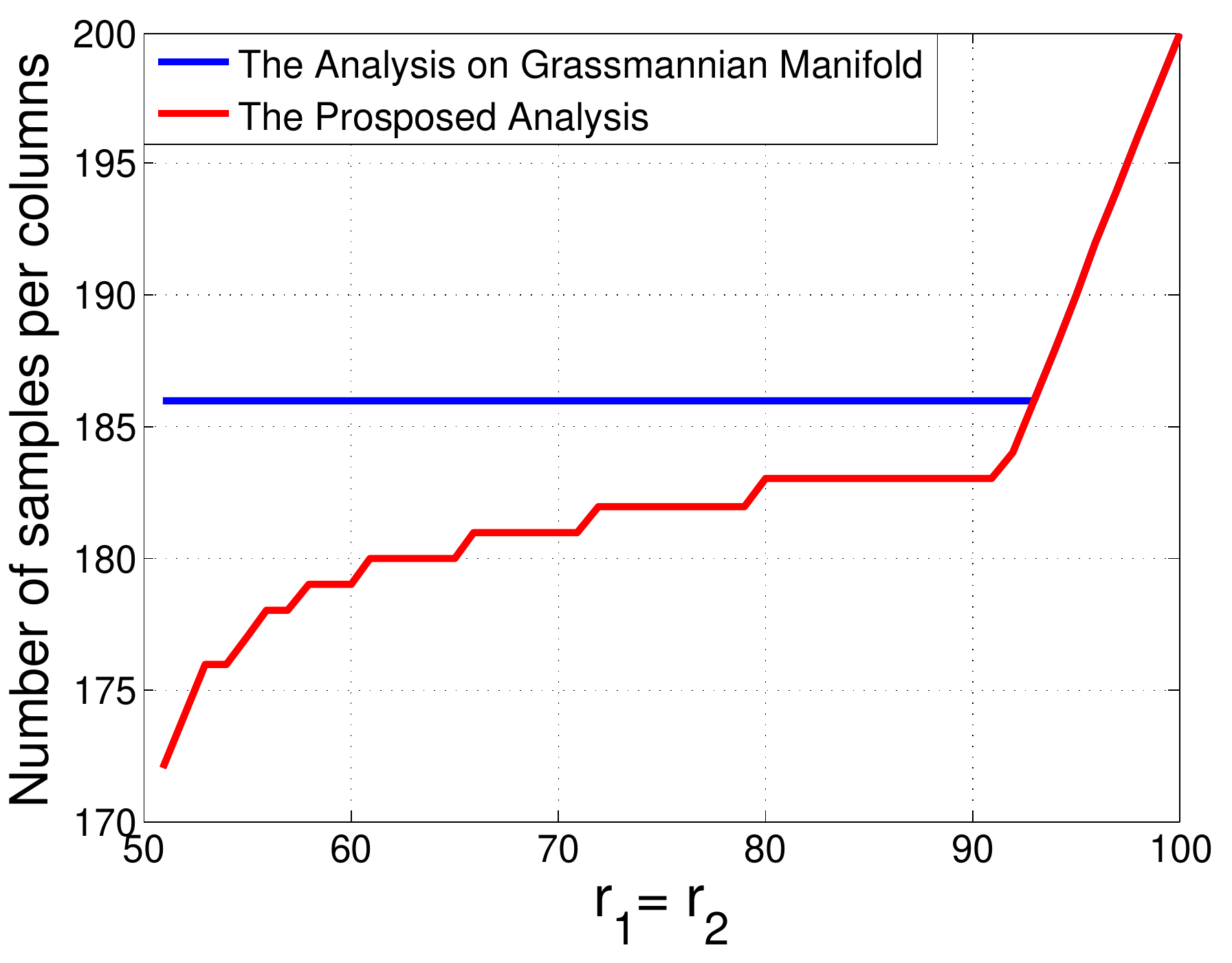}
\label{fig:subfigS4.3}
}
\caption[Optional caption for list of figures]{Lower bounds on the number of samples per column.}
\label{fignewcomp}
\end{figure}

\section{Conclusions}\label{concsec}

This paper characterizes fundamental algorithm-independent conditions on the sampling pattern for finite completability of a low-rank multi-view matrix through an algebraic geometry analysis on the manifold structure of multi-view data. A set of polynomials is defined based on the sample locations and we characterize the number of maximum algebraically independent polynomials. Then, we transform the problem of characterizing the finite or unique completability of the sampled data to the problem of finding the maximum number of algebraically independent polynomials among the defined polynomials. Using these developed tools, we have obtained the following results: (i) The necessary and sufficient conditions on the sampling pattern, under which there are only finite completions given the three rank constraints, (ii) Sufficient conditions on the sampling pattern, under which there exists only one completion given the three rank constraints, (iii) Lower bounds on the number of sampled entries per column that leads to finite/unique completability with high probability. 

\appendices
\section{Proof of Finite Completability for the Example in Section \ref{backg}}\label{app1}

Observe that Assumption $1$ holds, i.e., each column of $\mathbf{U}_1$ includes at least one observed entry and each column of $\mathbf{U}_2$ includes at least two observed entries. According to the definition of the constraint matrix, we have $\mathbf{\breve{\Omega}} = [\mathbf{\breve{\Omega}}_1|\mathbf{\breve{\Omega}}_2]$, where 

$$
\mathbf{\breve{\Omega}}_1 = \left[
\begin{array}{ cc }
1 & 1 \\ 
1 & 0 \\ 
0 & 1 \\ 
0 & 0 
\end{array}
\right], \  \ 
\text{and} \ \  \ 
\mathbf{\breve{\Omega}}_2 = \left[
\begin{array}{ ccc }
1 & 1 & 1 \\ 
1 & 1 & 1 \\ 
1 & 0 & 0 \\ 
0 & 1 & 1 
\end{array}
\right].
$$

Note that $r_1^{\prime} = r - r_2 = 0$, $r_2^{\prime} = r - r_1 = 1$ and $r^{\prime} = r_1 + r_2 - r= 1$. As a result, $nr -r^2 -r_1^2 -r_2^2 + r(r_1 +r_2) = 5$ and $\mathbf{\breve{\Omega}}$ has exactly $5$ columns. Suppose that $\mathbf{\Omega}^{\prime}$ is an arbitrary submatrix of $\mathbf{\Omega}$. In order to show finite completability of $\mathbf{U}$, it suffices to show \eqref{charfinitecom} holds. Let $\mathbf{\Omega}^{\prime}_1$ and $\mathbf{\Omega}^{\prime}_2$ denote the submatrix that consists of columns of $\mathbf{\breve{\Omega}}^{\prime}$ that correspond to the first view and second view, respectively. Note that $\mathbf{\breve{\Omega}}^{\prime }=[\mathbf{\breve{\Omega}}^{\prime }_1|\mathbf{\breve{\Omega}}^{\prime }_2]$. Therefore, we only need to verify 
\begin{eqnarray}\label{examplefinpro}
(g(\mathbf{\breve{\Omega}}^{\prime }_2)-2)^+ +  (g(\mathbf{\breve{\Omega}}^{\prime  })-1)^+ \geq c(\mathbf{\breve{\Omega}}^{\prime }).
\end{eqnarray}

There are $3$ different cases as follows:
\begin{enumerate}
\item $g(\mathbf{\breve{\Omega}}^{\prime }_2) = 0$: In this case, \eqref{examplefinpro} reduces to $(g(\mathbf{\breve{\Omega}}^{\prime  }_1)-1)^+ \geq c(\mathbf{\breve{\Omega}}^{\prime }_1)$. This is easy to verify by checking each sub-case that $\mathbf{\breve{\Omega}}^{\prime }_1$ has one or two columns of $\mathbf{\breve{\Omega}}_1$.
\item $g(\mathbf{\breve{\Omega}}^{\prime }_2) = 3$: In this case, \eqref{examplefinpro} reduces to $1+(g(\mathbf{\breve{\Omega}}^{\prime  })-1)^+ \geq c(\mathbf{\breve{\Omega}}^{\prime })$. We consider the following two sub-cases: 
\begin{itemize}
\item $\mathbf{\breve{\Omega}}^{\prime }_2$ is the first column of $\mathbf{\breve{\Omega}}_2$: Observe that in this case $c(\mathbf{\breve{\Omega}}^{\prime })= c(\mathbf{\breve{\Omega}}^{\prime }_1) +1$, and also we always have $g(\mathbf{\breve{\Omega}}^{\prime  }) \geq g(\mathbf{\breve{\Omega}}^{\prime  }_1)$. Hence, similar to the previous scenario, it suffices to show that $(g(\mathbf{\breve{\Omega}}^{\prime  }_1)-1)^+ \geq c(\mathbf{\breve{\Omega}}^{\prime }_1)$ which is easy to verify.
\item $\mathbf{\breve{\Omega}}^{\prime }_2$ does not include the first column of $\mathbf{\breve{\Omega}}_2$: Note that in this case $c(\mathbf{\breve{\Omega}}^{\prime }) \leq c(\mathbf{\breve{\Omega}}^{\prime }_1) +2$, and therefore it suffices to show that $(g(\mathbf{\breve{\Omega}}^{\prime  })-1)^+ \geq c(\mathbf{\breve{\Omega}}^{\prime }_1)+1$. This is easy to verify by considering the fact that in this case $g(\mathbf{\breve{\Omega}}^{\prime  })=4$ if and only if $\mathbf{\breve{\Omega}}^{\prime }_1$ includes the second column of $\mathbf{\breve{\Omega}}_1$, and $g(\mathbf{\breve{\Omega}}^{\prime  })=3$ otherwise.
\end{itemize}
\item $g(\mathbf{\breve{\Omega}}^{\prime }_2) = 4$: In this case, \eqref{examplefinpro} reduces to $2+(g(\mathbf{\breve{\Omega}}^{\prime  })-1)^+ \geq c(\mathbf{\breve{\Omega}}^{\prime })$. Note that $g(\mathbf{\breve{\Omega}}^{\prime }_2) = 4$ results that $g(\mathbf{\breve{\Omega}}^{\prime }) = 4$, and therefore \eqref{examplefinpro} reduces to $5 \geq c(\mathbf{\breve{\Omega}}^{\prime })$ which clearly always holds.
\end{enumerate}

\bibliographystyle{IEEETran}
\bibliography{bib}

\end{document}